\documentclass[a4paper,11pt]{article}
\usepackage{authblk}
\usepackage[utf8]{inputenc}
\usepackage{geometry}
\usepackage{palatino}

\usepackage{cite}
\usepackage{ulem} 
\usepackage{amsmath,amssymb,amsfonts,amsthm}
\usepackage{mathrsfs}
\usepackage{enumitem} %for enumerate enviroment
\usepackage{tabularx} %for tables
\usepackage{arydshln}
\usepackage{xcolor}
\def\BibTeX{{\rm B\kern-.05em{\sc i\kern-.025em b}\kern-.08em
    T\kern-.1667em\lower.7ex\hbox{E}\kern-.125emX}}
\usepackage{bm}
\usepackage{thmtools}
\usepackage{tikz-cd}
\usepackage{pgfplots}
\pgfplotsset{compat=1.18}
\usepgfplotslibrary{fillbetween}
\usepackage{hyperref}
\hypersetup{
    colorlinks=true,
    linkcolor=blue,
    citecolor=violet   
    }
\usepackage[nameinlink]{cleveref}

\newcommand{\tabitem}{~~\llap{\textbullet}~~}

\newtheorem{theorem}{Theorem}

\newtheorem{definition}[theorem]{Definition}

\newtheorem{remark}[theorem]{Remark}
\newtheorem{example}{Example}

\usepackage{ragged2e}

\renewcommand{\qed}{\hfill\rule{1.8mm}{1.8mm}}
\renewcommand{\proof}{\noindent{\bf Proof.\ }}

\usepackage{fancyhdr}
\title{Quasi-twisted codes: decoding and applications in code-based cryptography}
\author{Bhagyalekshmy S.\thanks{bhagyalekshmy.s@students.iiserpune.ac.in}}
\author[2]{Rutuja Kshirsagar \thanks{rkshirsagar@fujitsu.com}}
\affil[1]{Department of Mathematics, IISER Pune}
\affil[2]{Fujitsu Research of America, Inc.}
\date{April 2025 }

\begin{document}
\maketitle

\begin{abstract}
Quasi-twisted (QT) codes generalize several important families of linear codes, including cyclic, constacyclic, and quasi-cyclic codes. Despite their potential, to the best of our knowledge, there exists no efficient decoding algorithm for QT codes. In this work, we propose a syndrome-based decoding method capable of efficiently correcting up to $\frac{d^*-1}{2}$ errors, where $d^*$ denotes an HT-like lower bound on the minimum distance of QT codes, which we formalize here. Additionally, we introduce a Niederreiter-like cryptosystem constructed from QT codes. This cryptosystem is resistant to some classical attacks as well as some quantum attacks based on Quantum Fourier Sampling.
\end{abstract}

\maketitle

\section{Introduction}
Post-quantum cryptography has garnered significant attention due to advances in quantum computing, particularly following the development of Shor’s algorithm for prime factorization \cite{Shor.factorization}. In response to the emerging threat that quantum computers pose to classical cryptographic systems, the National Institute of Standards and Technology (NIST) initiated a standardization process for post-quantum cryptographic algorithms in 2016. Among the families of cryptographic systems under consideration, code-based cryptography \cite{CBC.HQC} is one of the few that advanced to the final stages of this process.

One of the earliest and most studied code-based schemes is the McEliece cryptosystem \cite{McE}, which was shown to be resistant to quantum attacks. A related variant, the Niederreiter cryptosystem \cite{Nied}, was introduced soon thereafter. Both rely on binary Goppa codes as the underlying structure. The security of these systems hinges on the difficulty of distinguishing Goppa codes from random linear codes — a problem that is known to be NP-hard \cite{Dec.random.code}. While these systems offer fast encryption and decryption, their practicality is hindered by the extremely large public and private key sizes.

Comprehensive introductions to code-based cryptography, including its mathematical foundations and security assumptions, are available in \cite{Intro} and \cite{Survey}.

Recent research efforts have aimed at replacing binary Goppa codes with alternative code families to reduce key sizes while preserving security. The central challenge in this approach is to ensure that the new code family maintains two key properties: the difficulty of decoding a random error, assuming a random code and the resistance to structural attacks that might reveal the algebraic properties of the code. Alternative codes explored in the literature include non-binary Goppa codes \cite{berstein.nb.goppa}, Reed-Solomon and generalized Reed-Solomon (GRS) codes \cite{baldi.RS, baldi.RS.patent, baldi.GRS, bolkema.GRS, khaturia.2RS, ivanov.GRS, khaturia.exp.RS, Tw-RS}, algebraic geometry codes \cite{janwa.AG, Tw-Herm}, Reed-Muller codes \cite{sidelnikov.RM}, low- and moderate-density parity check (LDPC and MDPC) codes \cite{baldi.QC.LDPC, misoczki.MDPC}, convolutional codes \cite{londahl.conv}, and cyclic or quasi-cyclic codes \cite{Upen-Pap, ye.cyclic, aguilar.QC}, among others.

However, many of these alternatives have been shown to be vulnerable to a variety of structural and decoding-based attacks \cite{couvreur.dist.attack, couvreur.subcode.attack, couvreur.wild.attack, couvreur.BBCRS.attack, gauthier.dist.attack, landais.attack, lavauzelle.TRS.attack, minder.attack, otmani.QC.attack, sidelnikov.GRS.attack, Wiesch, SS-attack, Sec-RS}. These vulnerabilities have highlighted the delicate balance between structure and security in code-based cryptography.

A particularly important quantum threat vector is Quantum Fourier Sampling (QFS), which underpins numerous quantum algorithms including Shor’s factorization, the discrete logarithm problem \cite{Shor.factorization}, and Simon’s algorithm \cite{simon.quantum}. As such, demonstrating that a cryptosystem can resist QFS-based attacks is a critical benchmark of quantum security. In this context, the Niederreiter-like cryptosystem based on quasi-cyclic codes \cite{Upen-Pap} has shown resilience against QFS-based attacks, offering both quantum security and practical advantages such as improved transmission and encryption rates, and notably smaller key sizes.

Building on this line of research, we propose a Niederreiter-like cryptosystem using quasi-twisted (QT) codes \cite{QT-codes} as the underlying structure. QT codes generalize quasi-cyclic, cyclic, and constacyclic codes, offering additional flexibility and potentially better cryptographic parameters. We demonstrate that our proposed system is secure against known QFS-based attacks, as well as against several prominent classical attacks.

A key component in designing a robust and efficient code-based cryptosystem is the availability of a reliable decoding algorithm tailored to the chosen code family. To this end, we present an efficient syndrome-based decoding algorithm for QT codes. This algorithm is guided by a Hartmann-Tzeng (HT)-like bound on the minimum distance of QT codes, $d^*$, which ensures the correction of $\varepsilon = \frac{d^* -1}{2}$ errors. The algorithm operates with a computational complexity quadratic in the code length, making it suitable for practical implementation without exposing structural vulnerabilities.

\subsection*{Outline.} 

The rest of the paper is organized as follows. Section \ref{sec:prelim} contains  background that is necessary to prove our results. In section \ref{sec:quasi-twisted codes}, we formalize a lower bound (HT-like bound) on the minimum distance of the quasi-twisted codes, $d^*$ and propose a syndrome-based decoding algorithm for a length $n$ QT code that corrects up to in $\varepsilon = \frac{d^* -1}{2}$ errors in $\mathcal{O}(n^2)$ time. In section \ref{sec:Niederreiter cryptosystem}, we propose a Niederreiter-like cryptosystem based on QT codes and analyze its security against some classical as well as quantum attacks.

\section{Preliminaries}\label{sec:prelim}

Given a positive integer $t$, $[t]:=\{ 0, 1,  \dots, t-1 \}$. Let $\mathbb{F}_q$ be a field where $q = p^s$ is a power of some prime, $p$ and a vector space over this field, $\mathbb{F}_q^n$ for some positive integers $s,n.$ Let $\mathbb{F}_q^*$ denote the set of non-zero elements of $\mathbb{F}_q$. The set of $m \times n$ matrices with entries in $\mathbb{F}_q$ is denoted by $\mathbb{F}_q^{m \times n}$. The entry in the $i^{th}$ row and $j^{th}$ column of a matrix $A \in \mathbb{F}_q^{m \times n}$ is denoted by $A_{ij}$. Let $\mathcal{C}$ is an $[n,k,d]_q$ linear code with length $n$, dimension $k$ and minimum distance $d$ over $\mathbb{F}_q$. The weight of a codeword, $wt(c)$ is the (Hamming) distance between the codeword, $c$ and the zero codeword. Let $G \in \mathbb{F}_q^{k \times n}$ and $H \in \mathbb{F}_q^{n-k \times n}$ denote the generator matrix and parity check matrix of $\mathcal{C}$, respectively. Given an $[n,k]$ linear code $\mathcal{C},$ the subspace of $\mathbb{F}_q^n$ containing all those vectors that are orthogonal to every codeword in $\mathcal{C}$ forms the \textit{dual code, $\mathcal{C}^\perp$} of code $\mathcal{C}$. That is, 
$$\mathcal{C}^\perp := \left\{\mathbf{u} \in \mathbb{F}_q^n: \mathbf{u\cdot v} =0, \forall\mathbf{v} \in \mathcal{C}\right\}.$$

We now turn our attention to cyclic codes and its variants. 

\begin{definition}\label{def.cyclic}{\bf (Cyclic codes)}
Suppose $\mathcal{C}$ is an $[n, k, d]$ linear code. Let $\lambda \in \mathbb{F}_q^*$. Suppose $c = (c_0, c_1, c_2, \ldots, c_{n-1}) \in \mathcal{C}$ be a codeword. Then $\mathcal{C}$ is a \textit{cyclic code} if a shift of $c$ by $1$ position,
\[(c_{n-1}, c_0, \ldots, c_{n-2}) \]
is also a codeword in $\mathcal{C}$. 
\end{definition}

Any codeword $c = (c_0, c_1, c_2, \ldots, c_{n-1})$ in a cyclic code, $\mathcal{C}$ can be denoted using a polynomial $c(X) = c_0 + c_1 X + \cdots + c_{n-1} X^{n-1} \in \mathbb{F}_q[X]$, called its generating function. Furthermore, among codewords of  $\mathcal{C}$ there exists a unique codeword, $g$ such that the degree of the associated polynomial, $g(X) \in \mathbb{F}_q[X]$ is minimal. Every polynomial $c(X)$ corresponding to a codeword of $\mathcal{C}$ can be obtained from $g(X) $ by multiplication with a polynomial $a(X) \in \mathbb{F}_q[X]$: $c(X) = a(X)g(X)$. The polynomial $g(X)$ is called the generating polynomial of $\mathcal{C}$. 

The generator polynomial of a cyclic code has its roots as powers of a field element. Based on the patterns in the set of these powers and the length of the code, a lower bound on the minimum distance of cyclic codes was proposed independently in \cite{BCH-1} and \cite{BCH-2}. This bound is known as the BCH bound. A generalization to this BCH bound was given by Hartmann and Tzeng in \cite{HT}, known as the HT-bound.

\begin{definition}{\bf (HT-Bound)}\cite[Thm. 2]{HT}\label{thm.HT}
    Consider a cyclic code, $\mathcal{C}$ of length $n$ and a generator polynomial, $g(X)$. Let $\beta \in \mathbb{F}_q^m$ be a non-zero element of order $n$. Let $\delta \geq 2, a, n_1, n_2, s $ be positive integers such that $gcd(n,n_1) = 1$ and $gcd(n,n_2) = 1.$ If $g(\beta^{a + in_1 + jn_2}) = 0$ for $i = 0, 1, \ldots, \delta-2$ and $j = 0, 1, \ldots, s,$ then the minimum distance of $\mathcal{C}$ is $d \geq \delta + s$.
\end{definition} 

We now restate the definition of quasi-twisted codes as defined by Y. Jia in \cite{QT-codes}, which are the primary focus of our work. 

\begin{definition}\label{def.QT}{\bf (Quasi-twisted codes)}
Suppose $\mathcal{C}$ is an $[n, k]$ linear code. Let $\lambda \in \mathbb{F}_q^*$ and $\ell$ be a positive integer. Suppose $c = (c_0, c_1, c_2, \ldots, c_{n-1}) \in \mathcal{C}$ is a codeword. Then $\mathcal{C}$ is a \textit{$(\lambda,\ell)$-quasi-twisted (QT) code} if a constacyclic shift of $c$ by $\ell$ positions given by,
\[(\lambda c_{n-\ell}, \lambda c_{n-\ell+1}, \ldots, \lambda c_{n-1}, c_0, \ldots, c_{n-\ell-1}) \]
is also a codeword in $\mathcal{C}$. (here, the subscripts are all taken modulo $n$.) %Quasi-twisted codes were introduced in \cite{QT-codes}. 
\end{definition}

Assume $\ell$ always divides $n$. Let $m = n/\ell$. Code $\mathcal{C}$ then corresponds to a submodule of $\left({\mathbb{F}_{q}[X]}/{\left<X^m-\lambda\right>}\right)^\ell$ in ${\mathbb{F}_{q}[X]}/{\left<X^m-\lambda\right>}.$ Cyclic codes, constacyclic codes and quasi-cyclic codes are all special cases of quasi-twisted codes. Consider a $(\lambda,\ell)$-quasi-twised code, if $\lambda = 1$ and $\ell=1$, then it is a cyclic code. Substituting $\lambda = 1$ or $q=2$, gives an $\ell$-quasi-cyclic code and substituting $\ell =1$, gives a $\lambda$-constacyclic code. Throughout the discussion of this paper, we focus on the case when $gcd(m,p) = 1,$ where $p$ is the characteristic of $\mathbb{F}_q.$ 

\subsection{Spectral theory of quasi-twisted codes}

In this section, we reiterate some important concepts from spectral theory of quasi-twisted codes that are necessary to prove some of our results. A detailed explanation of the same can be found in \cite{Spec-QT}.

\subsubsection{Reduced Gröbner basis}

Consider a finite field, $\mathbb{F}_q$ and a polynomial ring over this field, $\mathbb{F}_q[X].$ Let $\mathcal{C}$ denote an $[n=m\ell, k, d]_q$ $(\lambda,\ell)$-quasi-twisted code. Any codeword, $c \in \mathcal{C}$ can be written as an $m \times \ell$ matrix in $\mathbb{F}_q^{n}$ as given in equation \eqref{eq:codeword_QT_code}.
\begin{equation}
\label{eq:codeword_QT_code}
    c = 
    \begin{bmatrix}
        c_{0,0} & c_{0,1} & \hdots & c_{0,\ell-1}\\
        \vdots & \vdots & \ddots & \vdots \\
        c_{m-1,0} & c_{m-1,1} & \hdots & c_{m-1,\ell-1} \\
    \end{bmatrix}.
\end{equation}
Each column of this matrix can be mapped to a polynomial, $c_i(X) = \sum\limits_{j=0}^{m-1}c_{j,i}X^j \in \mathbb{F}_q[X]/{\left<X^m-\lambda\right>}$ for all $i \in [\ell]$ . Let $R$ denote $\mathbb{F}_q[X]/{\left<X^m-\lambda\right>}$.
The $R$-module isomorphism, $\phi$ in  equation \eqref{eq.4.1} maps the matrix corresponding to a codeword to a polynomial in $R^\ell$. 
\begin{equation}\label{eq.4.1}
\begin{aligned}
    & \phi :  \hspace{20mm}\mathbb{F}_q^{n} & \longrightarrow & R^\ell \\
    & c = 
    \begin{bmatrix}
        c_{0,0} & c_{0,1} & \hdots & c_{0,\ell-1}\\
        \vdots & \vdots & \ddots & \vdots \\
        c_{m-1,0} & c_{m-1,1} & \hdots & c_{m-1,\ell-1} \\
    \end{bmatrix}
    & \longmapsto & c(X) \\ %\equalto{c(X)}{(c_0(X),c_1(X), \hdots, c_{\ell-1}(X))} \\
\end{aligned}
\end{equation}
Here, $c(X)= (c_0(X),c_1(X), \hdots, c_{\ell-1}(X)) \in R^\ell$, where $c_i(X) = \sum\limits_{j=0}^{m-1}c_{j,i}X^j $ for all $i \in [\ell]$. The polynomials, $c_i(X)$ are closed under multiplication by $X$ and reduction modulo $X^{m}-\lambda$ for all $i \in [\ell]$. Therefore, the invariance of codewords in $\mathbb{F}_q^{m\ell}$ under $\lambda$-constashift by $\ell$ positions (row rotation by $\ell$ positions such that the rotated positions have weight $\lambda$), now corresponds to a matrix in $\mathbb{F}_q^{m\times l}$ closed under the row $\lambda$-constashift. That means, a $(\lambda,\ell)$-quasi-twisted code can be seen as an $R$-submodule of $R^\ell.$

M. F. Ezerman et al. has shown in \cite{Spec-QT} that a generating set of quasi-twisted codes can also be expressed in the form of a \textit{reduced Gröbner basis} with respect to the position-over-term order, represented in equation \eqref{QT.Gröbner.basis}. 

\begin{equation}
    \Tilde{G}(X) = \begin{pmatrix}\label{QT.Gröbner.basis}
        g_{0,0}(X) & g_{0,1}(X) & \hdots & g_{0,\ell-1}(X) \\
        & g_{1,1}(X) & \hdots & g_{1,\ell-1}(X) \\
        & & \ddots & \vdots\\
        \hspace{17mm} \smash{\text{\huge 0}} & & & g_{\ell-1,\ell-1}(X) \\
    \end{pmatrix}
\end{equation}

%\vspace{-4mm}
For every $(\lambda, \ell)$-quasi-twisted code, its reduced Gröbner basis $\Tilde{G}(X)$ satisfies the following properties. 
\begin{enumerate}
    \item $g_{i,j}(X) = 0, \hspace{2mm} \forall
    0 \leq j < i < \ell$.
    \vspace{-2mm}
    \item $deg(g_{i,j}(X)) < deg(g_{j,j}(X)), \hspace{2mm} \forall i<j.$
    \vspace{-2mm}
    \item \label{prop3} $g_{i,i}(X)|(X^m-\lambda), \hspace{2mm} \forall i \in [l].$
    \vspace{-2mm}
    \item If $g_{i,i}(X) = X^m-\lambda, \text{ then } g_{i,j}(X) = 0 \hspace{2mm} \forall j \neq i. $
\end{enumerate}

Every codeword polynomial $c(X) \in R^\ell$ can then be written as $c(X) = a(X)\Tilde{G}(X)$, where $a(X) = a_0 + a_1X + \cdots + a_{m-1}X^{m-1}$ in the ideal $R$. The $\mathbb{F}_q$-dimension of $\mathcal{C}$ is, 
\begin{equation}
    k = m\ell - \sum\limits_{i=0}^{\ell-1}deg\big(g_{i,i}(X)\big) = \sum\limits_{i=0}^{\ell-1}\Big(m-deg\big(g_{i,i}(X)\big)\Big)
\end{equation}
The determinant of $\Tilde{G}(X)$ is $det\big(\Tilde{G}(X)\big) = \prod\limits_{i=0}^{\ell-1}g_{i,i}(X)$, which means that every eigenvalue,  $\beta_i$ of the quasi-twisted code $\mathcal{C}$, is a root of the polynomial, $g_{i,i}(X)$ for some $i \in [\ell].$ From property \ref{prop3} of the reduced Gröbner basis, for all $i \in [m]$ each eigenvalue, $\beta_i = \alpha\xi^i,$ where $\alpha$ is a fixed $m^{th}$ root of $\lambda$ and $\xi$ is a primitive $m^{th}$ root of unity.

The \textit{algebraic multiplicity} of an eigenvalue, $\beta_i$ is defined as the largest integer, $a$ such that $(X-\beta_i)^a|det\big(\Tilde{G}(X)\big)$. Its \textit{geometric multiplicity} is the dimension of the null space of $\Tilde{G}(\beta_i)$, which is the eigenspace of $\beta_i$. That is, $\mathcal{V}_{i} := \{\mathbf{v}\in \mathbb{F}^\ell : \Tilde{G}(\beta_i)\mathbf{v}^\intercal = \mathbf{0}\}$. Here $\mathbb{F}$ is the splitting field of $X^m - \lambda$ (the smallest extension of $\mathbb{F}_q$ containing all the roots of $X^m-\lambda$). The algebraic multiplicity of an eigenvalue of a  $(\lambda,\ell)$-quasi-twisted code is equal to its geometric multiplicity \cite{Spec-QT}.

Consider an intersection of eigenspaces, $\mathcal{V}$ corresponding to a subset of eigenvalues. The minimum distance of a quasi-twisted code is then given by the minimum distance of the cyclic code, which has a generator polynomial with these eigenvalues as its roots. This forms the motivation behind the HT-like bound, discussed in theorem \ref{HT-like Bound-Th}.

\begin{definition}{\bf (Eigencode)}\label{eigencode}
    Consider an eigenspace, $\mathcal{V} \subseteq \mathbb{F}_q^\ell.$ An $[\ell, k_\mathbb{C},d_\mathbb{C}]_q$ eigencode $\mathbb{C}$ on $\mathcal{V}$ is defined in equation \eqref{eq:eigencode},
    \begin{equation}
    \label{eq:eigencode}
        \mathbb{C} := \Bigg\{c \in \mathbb{F}_q^\ell : \forall \mathbf{v} \in \mathcal{V}, \sum\limits_{i=0}^{\ell-1}v_ic_i = 0\Bigg\}.
    \end{equation}
    If for some vector $\mathbf{v} \in \mathcal{V},$ its elements $v_1, v_2, \hdots, v_{\ell-1}$ are linearly independent over $\mathbb{F}_q$, then the corresponding eigencode is $\mathbb{C} = {(0,0,\hdots,0)}$ and its minimum distance
    $d_\mathbb{C}$ is assumed to be infinity.
\end{definition}

\subsubsection{Spectral bound on minimum distance of quasi-twisted codes} 
Consider an $[m\ell, k, d]_q$ $(\lambda ,\ell)$-quasi-twisted code, $\mathcal{C}$. Let $\xi$ be a primitive $m^{th}$ root of unity and $\alpha$ be an $m^{th}$ root of $\lambda$. Let $\Omega := \{ \alpha \xi ^i : 0 \leq i \leq m-1 \}$ form the set of all $m^{th}$ roots of $\lambda$ or equivalently the set of roots of $X^m - \lambda$. Recall that for $\ell = 1,$ we get a $\lambda$-constacyclic code. Given the generator polynomial, $g(X)$ of a $\lambda$-constacyclic code, the set $L := \{\alpha \xi ^i : g(\alpha \xi ^i) = 0\}$ is called its \textit{defining set}. %\RK{why should the defining set be $\ell$? Notation is confusing to reader.}\BS{addressed}\\ 

Let $\mathbb{C}$ be an eigencode corresponding to an eigenspace $\mathcal{V}$ as stated in definition \ref{eigencode}. We now re-state a spectral bound on the minimum distance of quasi-twisted codes, which is presented in \cite[Theorem 11]{Spec-QT}. We improve this bound, as summarized in theorem \ref{HT-like Bound-Th}.
\begin{theorem}\cite[Theorem 11]{Spec-QT}\label{spec.bound.QT}
    Consider a $(\lambda,\ell)$-quasi-twisted code $\mathcal{C}$. Let $\overline{\Omega} \subset \Omega$ denote the non-empty set of eigenvalues of $\mathcal{C}.$ Consider a non-empty subset of eigenvalues, $P \subseteq \overline{\Omega}$ and a $\lambda$-constacyclic code $\mathcal{C}_P$ with its defining set $L \supseteq P.$ Let $d_P$ denote a lower bound on the minimum distance of $\mathcal{C}_P.$ Consider the intersection of eigenspaces ($\mathcal{V}_\beta$) of the eigenvalues ($\beta$) in $P,$ $\mathcal{V}_P := \bigcap\limits_{\beta \in P}\mathcal{V}_\beta$ and the corresponding eigencode $\mathbb{C}_P.$ Then, the minimum distance of the $(\lambda,\ell)$-quasi-twisted code $\mathcal{C}$ is given by, $d_\mathcal{C} \geq min\{d_P,d_{\mathbb{C}_P}\}.$
\end{theorem}

\subsection{Niederreiter cryptosystem}\label{Nied.desc.}
%\RK{shorten the desciption, add a diagram rather than explanation and table.}

This cryptosystem was introduced by H. Niederreiter in \cite{Nied}. It has three sub-algorithms: key generation, encryption, and decryption. Figure \ref{fig:Niederreiter} summarizes the Niederreiter cryptosystem.

\begin{figure}[h!]
    \centering
    \includegraphics[width=1\linewidth]{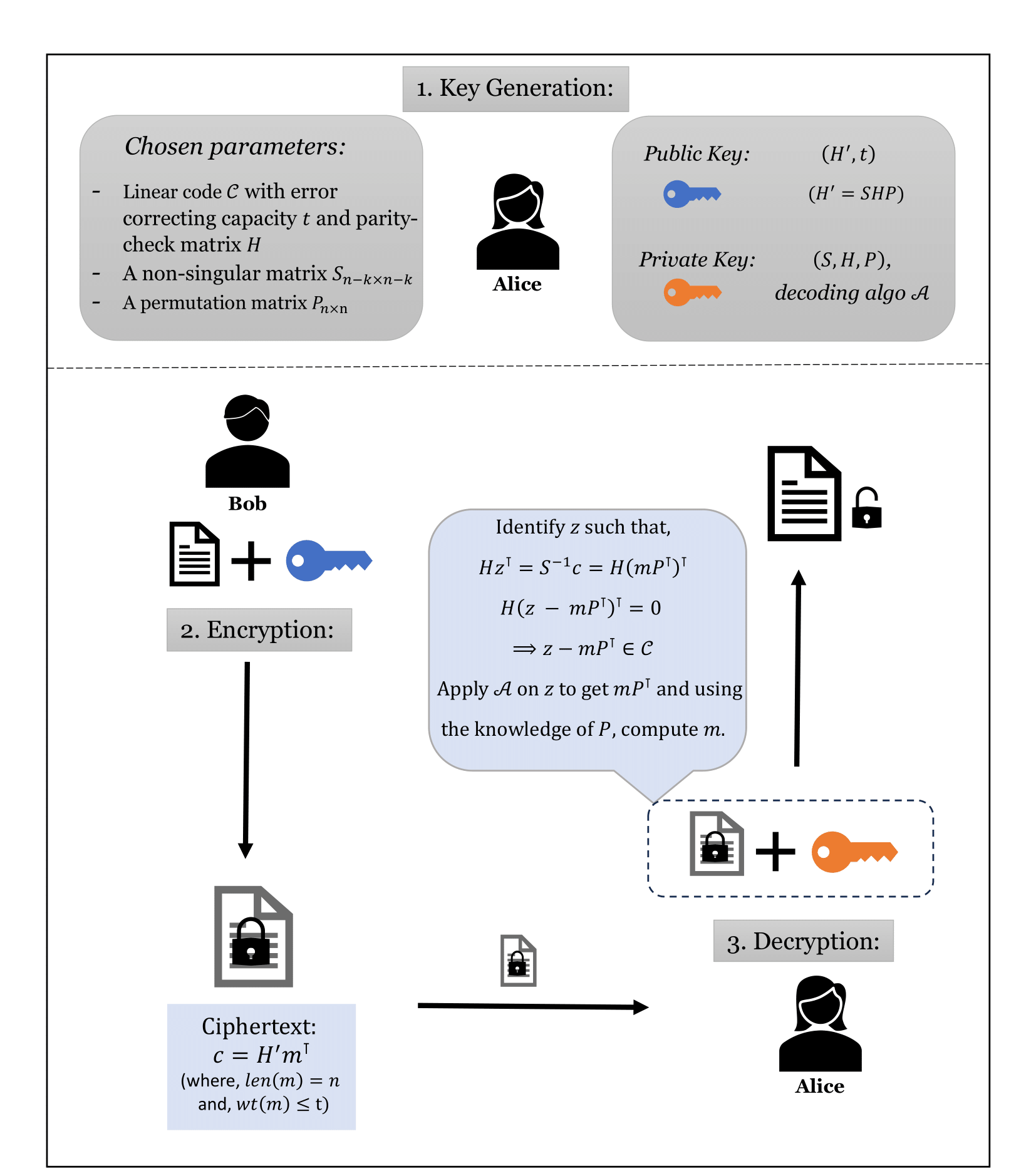}
    \caption{Niederreiter Cryptosystem: Alice chooses a linear code, $\mathcal{C}$ with a parity check matrix, $H$ and an efficient decoding algorithm, $\mathcal{A}$ that can correct up to $t$ errors. She also chooses two other matrices $S$ and $P$ as described above. She generates her public key, $(H'= SHP,t)$ and her private key, $(S, H, P)$ and algorithm $\mathcal{A}$. She publishes her public key, which is used by Bob to encrypt his message, $m$. Bob sends the ciphertext, $c = H'm^\intercal$. Alice decrypts it back to the intended message, using her private key.} 
    \label{fig:Niederreiter}
\end{figure}

\section{Decoding quasi-twisted codes}\label{sec:quasi-twisted codes}
In this section, we formulate a Hartmann-Tzeng(HT)-like bound on the minimum distance of quasi-twisted codes, $d^* $. %$= min\big(\delta + s, d_\mathbb{C}\big)$. 
Furthermore, we describe a syndrome-based decoding algorithm of quasi-twisted codes that corrects up to $\varepsilon \left(\leq \frac{d^*-1}{2}\right)$ errors in $\mathcal{O}(n^2)$ time. Here, $n$ is the length of the quasi-twisted code. 

\subsection{HT-like bound on minimum distance of quasi-twisted codes}\label{Bound on QT}
%\vspace{-2mm}
The HT bound on the minimum distance of cyclic codes is presented in definition \ref{thm.HT}. Similar bounds have been defined for constacyclic as well as quasi-cyclic codes. Here, we present an HT-like bound on the minimum distance of quasi-twisted codes. This HT-like bound can be deduced from \cite[Remark 5]{Spec-QT}, following the motivation behind the HT-like bound for quasi-cyclic codes (\cite[Theorem 1]{Dec-QC}). However, to the best of our knowledge, this is the first instance of a formal proof for this bound.
%however, to the best of our knowledge, this lower bound is not formally written anywhere. 

\begin{theorem}[\textbf{HT-like Bound}]\label{HT-like Bound-Th} 
Consider an $[m\ell, k, d]_q$ $(\lambda ,l)$-quasi-twisted code $\mathcal{C}$. For positive integers $ a, n_1, n_2, s$ and  $\delta \geq 2$ with $gcd(m,n_1) = 1$ and $gcd(m,n_2) < \delta,$ define the set,
\begin{align*}
    D := & \{ a, a+n_1, \ldots, a+ (\delta -2)n_1, \\
    & a + n_2, a+ n_2 +n_1, \ldots, a+n_2 + (\delta -2)n_1, \\
    & \hspace{6mm} \vdots \hspace{15mm} \vdots \hspace{8mm} \hdots \hspace{11mm} \vdots \\
    & a + sn_2, a+ sn_2 +n_1, \ldots, a+sn_2 + (\delta -2)n_1\}. \\
\end{align*}
%$D := \{ a, a+n_1, \ldots, a+ (\delta -2)n_1,$ \\
%\hspace{11mm} $a + n_2, a+ n_2 +n_1, \ldots, a+n_2 + (\delta -2)n_1,$\\
%\hspace{17mm}$\vdots$ \hspace{15mm} $\vdots$ \hspace{8mm} $\hdots$ \hspace{11mm} $\vdots$\\
%\hspace{13mm}$a + sn_2, a+ sn_2 +n_1, \ldots, a+sn_2 + (\delta -2)n_1\}$\\
 
For some $a \geq 0$, the eigenvalues of the quasi-twisted code $\mathcal{C}$ are $\beta_i = \alpha \xi ^i, \forall i \in D$. 
Let the corresponding eigenspaces be denoted by $\mathcal{V}_i, \forall i \in D.$ Consider the intersection of these eigenspaces, $\mathcal{V} := \bigcap\limits_{i\in D} \mathcal{V}_i$ and an eigenvector $\mathbf{v} = (v_0, v_1, \ldots, v_{\ell-1}) \in \mathcal{V}.$
Let the eigencode given by $\mathcal{V}$ (that is, the direct sum of eigencodes defined by $\mathcal{V}_i$) be $\mathbb{C}$ with minimum distance $d_\mathbb{C}$.
If 
\vspace{.5mm}
\begin{equation}\label{Th-eqn}
    c(\alpha \xi^{a+i_1n_1+i_2n_2})\mathbf{v}^\intercal = 0, \hspace{2mm} \forall i_1 \in [\delta -1], i_2 \in [s+1] 
\end{equation}
\vspace{.5mm}
holds true %for an eigenvector $\textbf{v} = (v_0, v_1, \ldots, v_{\ell-1}) \in \mathcal{V}$ and 
for all $c(X) = \big(c_0(X), c_1(X), \ldots, c_{\ell-1}(X)\big) \in R^\ell$, then the HT-like bound on the minimum distance, $d$ of the quasi-twisted code, $\mathcal{C}$ is 
\begin{equation}
\label{eq:HT.like.bound.QT}
d \geq d^* = min\big\{\delta + s, d_\mathbb{C}\big\}.
\end{equation}
\end{theorem}

\begin{proof}
Recall that an eigenspace, $\mathcal{V}_i$ of an eigenvalue, $\beta_i$ is defined as, 
\begin{equation}
    \mathcal{V}_{i} := \Big\{\mathbf{v}\in \mathbb{F}^\ell : \Tilde{G}(\beta_i)\mathbf{v}^\intercal = \mathbf{0}\Big\}.
\end{equation}
Since $\mathcal{V} := \bigcap\limits_{i\in D} \mathcal{V}_i$ is the intersection of the eigenspaces $\mathcal{V}_i$, we can conclude that 
\vspace{-1mm}
\begin{equation}\label{eqn5.2}
    \Tilde{G}(\beta_i)\mathbf{v}^\intercal = \mathbf{0}, \hspace{2mm} \forall i \in D, \forall \mathbf{v} \in \mathcal{V}.
\end{equation}
Recall that for any codeword, its corresponding codeword polynomial $c(X)\in R^\ell$ can be written as $c(X) = a(X)\Tilde{G}(X).$ Now, the LHS of condition \eqref{Th-eqn} can be simplified as follows. 
\vspace{-1mm}
\begin{equation}
\begin{split}
    c(\alpha \xi^{a+i_1n_1+i_2n_2})\mathbf{v}^\intercal = & c(\alpha \xi^{i})\mathbf{v}^\intercal, \hspace{2mm} \forall i \in D\\
    = & a(\alpha \xi^{i}))\Tilde{G}(\alpha \xi^{i})\mathbf{v}^\intercal, \hspace{2mm} \forall i \in D\\
    = & 0 \text{ (using equation \eqref{eqn5.2})}
\end{split}
\end{equation}
Let $\overline{\Omega}$ denote the set of all eigenvalues of the $(\lambda,\ell)$-quasi-twisted code $\mathcal{C}.$ The set ${\Omega} = \{\beta_i : i \in D\}$ is a subset of $\overline{\Omega}.$ In other words, %the premise of this theorem can be considered as a special case of theorem \ref{spec.bound.QT}, where $P = \overline{\Omega}.$ 
the premise of this theorem is same as that of theorem \ref{spec.bound.QT}. Therefore, the minimum distance of the $(\lambda,\ell)$-quasi-twisted code $\mathcal{C}$ is $d_\mathcal{C} \geq min\{d_{\Omega},d_{\mathbb{C}_{\Omega}}\},$ where we have the following observations on $d_{\Omega}$ and $d_{\mathbb{C}_{\Omega}}.$ 

\cite[Corollary 2.ii]{Spec-QT2} gives a bound on the minimum distance of a $\lambda$-constacyclic code. Therefore, we have $d_{\Omega} \geq \delta+s.$ 

Since ${\Omega} = \{\beta_i : i \in D\}, \mathbb{C}_{\Omega}$ is same as $\mathbb{C},$ the eigencode corresponding to the common eigenspace of eigenvalues in ${\Omega}$, the minimum distance of a $(\lambda,\ell)$-quasi-twisted code $\mathcal{C}$ is given by, $d \geq d^* = min\big\{\delta + s, d_\mathbb{C}\big\}.$
\qed

\subsection{A syndrome-based decoding algorithm}\label{subsec:dec-algo}

We begin by defining some polynomials and matrices that are necessary to develop our syndrome-based decoding algorithm for quasi-twisted codes.
    
Consider an $[m\ell, k, d]_q$ $(\lambda ,\ell)$-quasi-twisted code $\mathcal{C}$. 
Let $r(X) \in R^\ell$ be a polynomial corresponding to a received word, $r \in \mathbb{F}_q^n$ (under the $R$-module isomorphism defined in equation \eqref{eq.4.1}) and $e(X) \in R^\ell$ be a polynomial corresponding to an error vector, $e \in \mathbb{F}_q^n$. A received word, $r$ can be written as a summation of some codeword, $c$, and an error vector, $e$. Thus, we can write the corresponding polynomial, $r(X)$ as a sum of the corresponding codeword polynomial, $c(X)$ and the corresponding error polynomial, $e(X)$ as defined in equation \eqref{eq:rcd.word.poly}. 

\begin{equation}\label{eq:rcd.word.poly}
\begin{split}
    r(X) & = \big(r_0(X), r_1(X), \ldots, r_{\ell-1}(X) \big)\\
    & = \big(c_0(X)+e_0(X), c_1(X)+e_1(X), \ldots, c_{\ell-1}(X)+e_{\ell-1}(X) \big)\\
\end{split}
\end{equation}
Let $\mathcal{E} = \{i_1, i_2, \ldots, i_\varepsilon\}$ be a set of error locations. Note that $\mathcal{E}$ is also a union of all the error locations $\mathcal{E}_i$ in each polynomial $r_i(X) \in R$, as defined in equation \eqref{eq:err.loc.poly}.
\begin{equation}\label{eq:err.loc.poly}
    \mathcal{E} := \bigcup_{i=0}^{\ell-1}\mathcal{E}_i.
\end{equation}

We can therefore establish an upper bound on the number of error locations, $\varepsilon$. 
\begin{equation}
    |\mathcal{E}| = \varepsilon \leq \sum_{i=0}^{\ell-1}|\mathcal{E}_i|.
\end{equation}

For quasi-twisted code, 
\begin{equation}
    c_j(X) = \sum_{i \in [m]}c_{i,j}X^i \text{ and }
    e_j(X) = \sum_{i \in \mathcal{E}_i }e_{i,j}X^i, \hspace{2mm} \forall j \in [\ell].
\end{equation}

\newpage
\textbf{\underline{Syndrome Polynomials:}}

\vspace{2mm}
Let $\mathcal{V} = \bigcap\limits_{i_1\in [\delta-1],i_2\in [s+1]}\mathcal{V}_{a+i_1n_1+i_2n_2}$. Consider an eigenvector $\mathbf{v}=(v_0,v_1,\hdots,v_{\ell-1})\in \mathcal{V}$. The $(s+1)$ syndrome polynomials are defined as follows.
\begin{equation}
\begin{split}
    S_t(X) & := \sum\limits_{i=0}^{\infty}\left(\sum\limits_{j=0}^{\ell-1} r_j(\alpha \xi^{a+in_1+tn_2})v_j\right)X^i \hspace{2mm} \text{mod } X^{\delta-1} \\
    & = \sum\limits_{i=0}^{\delta-2}\left(\sum\limits_{j=0}^{\ell-1} r_j(\alpha \xi^{a+in_1+tn_2})v_j\right)X^i, \hspace{2mm} \forall t \in [s+1]. \label{S(X)}
\end{split}
\end{equation}
Using equation \eqref{Th-eqn} from theorem \ref{HT-like Bound-Th} for all $t \in [s+1]$, we can simplify $S_t(X)$.
\begin{equation}
    S_t(X) = \sum\limits_{i=0}^{\delta-2}\left(\sum\limits_{j=0}^{\ell-1} e_j(\alpha \xi^{a+in_1+tn_2})v_j\right)X^i.
\end{equation}

\textbf{\underline{Error locator polynomial:}}

\vspace{2mm}
We use the generalized Berlekamp-Massey algorithm \cite{Gen.BM-Algo} in the decoding procedure. Using this algorithm, we get the error locator polynomial to be a polynomial with $\xi^{-in_1}, \forall i \in \mathcal{E}$ as its roots. Therefore, we define the error locator polynomial as follows.
\begin{equation}\label{Lambda_poly}
    \begin{split}
        \Lambda(X) := & \prod_{i\in \mathcal{E}}(1-X\xi^{in_1})\\
        = & 1 - \Lambda_1X + \Lambda_2X^2 - \ldots \pm \Lambda_{\varepsilon}X^\varepsilon 
    \end{split}
\end{equation}

\textbf{\underline{Error evaluator polynomials:}}

\vspace{2mm}
For all $t \in [s+1]$, let $\Omega_t(X)$ denote error evaluator polynomials defined by a linear system of equations. 
\begin{equation}\label{Omega}
    \Lambda(X) \cdot S_t(X) \equiv \Omega_t(X) \hspace{2mm} \text{mod} \hspace{1mm} X^{\delta-1},\hspace{2mm} \forall t \in [s+1]
\end{equation}

\begin{remark}\label{Key-Eqns.}
    The $(s+1)$ equations in the linear system of equations \eqref{Omega} form the `Key Equations'. 
    Solving these equations jointly is equivalent to applying the generalized Berlekamp-Massey algorithm (\cite{Gen.BM-Algo}) 
    or the generalized expanded Euclidean algorithm (\cite{Gen.EA-Algo}) on $(X^{\delta-1}, S_0(X), S_1(X), \ldots, S_s(X))$. We use the former one for our analysis.
    %This gives us, as output, the error locator polynomial with roots $\xi^{-in_1}, \forall i \in \mathcal{E}$.
\end{remark}

Let us now look at a matrix representation of the syndrome polynomial and its decomposition carried out in accordance with a matrix operation $*$, which was introduced in \cite{*multi.}. A more detailed explanation can be found in \cite[Section VI.A]{Gen.BM-Algo}. 

The syndrome matrix, $S$ is defined as follows.
\vspace{-5mm}
\begin{equation}
\begin{split}
    & \\
    S := &  \big(S^{\left<0\right>} \hspace{2mm} S^{\left<1\right>} \hspace{2mm} . . . \hspace{2mm} S^{\left<s\right>}\big)^\intercal,\\
\end{split}    
\end{equation}
\begin{equation*}
%\begin{split}
    %& 
    \text{where} \hspace{2mm} S^{\left<t\right>} := \big( S_{i+j}^{\left<t\right>} \big) _{i \in [\delta-1-\varepsilon]}^{j \in [\varepsilon+1]}, \hspace{2mm} \forall t \in [s+1] %\\
    %& 
    \text{ and} \hspace{2mm} S_{k}^{\left<t\right>} = \sum_{j=0}^{\ell-1} r_j(\alpha\xi^{a+kn_1+tn_2})v_j, \hspace{2mm} \forall k \in [\delta - 1], t \in [s+1]. %\\
%\end{split}    
\end{equation*}
That is,
\begin{equation}\label{S}
    S = 
    \begin{bmatrix}
    S_{0}^{\left<0\right>} & S_{1}^{\left<0\right>} & \ldots & S_{\varepsilon}^{\left<0\right>}\\
    S_{1}^{\left<0\right>} & S_{2}^{\left<0\right>} & \ldots & S_{\varepsilon+1}^{\left<0\right>}\\
    \vdots & \vdots & \hdots & \vdots\\
    S_{\delta-2-\varepsilon}^{\left<0\right>} & S_{\delta-1-\varepsilon}^{\left<0\right>} & \ldots & S_{\delta-2}^{\left<0\right>}\\
    \hdashline[2pt/2pt]
    S_{0}^{\left<1\right>} & S_{1}^{\left<1\right>} & \ldots & S_{\varepsilon}^{\left<1\right>}\\
    S_{1}^{\left<1\right>} & S_{2}^{\left<1\right>} & \ldots & S_{\varepsilon+1}^{\left<1\right>}\\
    \vdots & \vdots & \hdots & \vdots\\
    S_{\delta-2-\varepsilon}^{\left<1\right>} & S_{\delta-1-\varepsilon}^{\left<1\right>} & \ldots & S_{\delta-2}^{\left<1\right>}\\
    \hdashline[2pt/2pt]
    \vdots & \vdots & \hdots & \vdots\\
    \vdots & \vdots & \hdots & \vdots\\
    \hdashline[2pt/2pt]
    S_{0}^{\left<s\right>} & S_{1}^{\left<s\right>} & \ldots & S_{\varepsilon}^{\left<s\right>}\\
    S_{1}^{\left<s\right>} & S_{2}^{\left<s\right>} & \ldots & S_{\varepsilon+1}^{\left<s\right>}\\
    \vdots & \vdots & \hdots & \vdots\\
    S_{\delta-2-\varepsilon}^{\left<s\right>} & S_{\delta-1-\varepsilon}^{\left<s\right>} & \ldots & S_{\delta-2}^{\left<s\right>}\\
    
    \end{bmatrix}_{(\delta-1-\varepsilon)(s+1)\times(\varepsilon+1)}
\end{equation}

Let $S = XY\Tilde{X}$ be the decomposition of $S$. Here, the matrix $X$ is defined in equation \eqref{eq:S.decomposition.X},

\begin{equation}
\label{eq:S.decomposition.X}
\begin{split}
    X := &  \big( X^{\left<0\right>} \hspace{2mm} X^{\left<1\right>} \hspace{2mm} . . . \hspace{2mm} X^{\left<s\right>} \big)^\intercal, \\
\end{split}    
\end{equation}
where $X^{\left<t\right>} := \big( \xi^{(n_1i+n_2t)j} \big)_{i \in [\delta-1-\varepsilon]}^{j \in \{1,2,\ldots,\varepsilon\}}$ for all $t \in [s+1].$ That is,
\begin{equation}\label{X}
    X = 
    \begin{bmatrix}
    1 & 1 & \ldots & 1\\
    \xi^{n_1} & \xi^{2n_1} & \ldots & \xi^{\varepsilon n_1}\\
    \vdots & \vdots & \hdots & \vdots\\
    \xi^{(\delta-2-\varepsilon)n_1} & \xi^{2(\delta-2-\varepsilon)n_1} & \ldots & \xi^{\varepsilon(\delta-2-\varepsilon)n_1}\\
    \hdashline[2pt/2pt]
    \xi^{n_2} & \xi^{2n_2} & \ldots & \xi^{\varepsilon n_2}\\
    \xi^{(n_1+n_2)} & \xi^{2(n_1+n_2)} & \ldots & \xi^{\varepsilon(n_1+n_2)}\\
    \vdots & \vdots & \hdots & \vdots\\
    \xi^{\big((\delta-2-\varepsilon)n_1+n_2\big)} & \xi^{2\big((\delta-2-\varepsilon)n_1+n_2\big)} & \ldots & \xi^{\varepsilon\big((\delta-2-\varepsilon)n_1+n_2\big)}\\
    \hdashline[2pt/2pt]
    \vdots & \vdots & \hdots & \vdots\\
    \vdots & \vdots & \hdots & \vdots\\
    \hdashline[2pt/2pt]
    \xi^{sn_2} & \xi^{2sn_2} & \ldots & \xi^{\varepsilon sn_2}\\
    \xi^{(n_1+sn_2)} & \xi^{2(n_1+sn_2)} & \ldots & \xi^{\varepsilon(n_1+sn_2)}\\
    \vdots & \vdots & \hdots & \vdots\\
    \xi^{\big((\delta-2-\varepsilon)n_1+sn_2\big)} & \xi^{2\big((\delta-2-\varepsilon)n_1+sn_2\big)} & \ldots & 
    \xi^{\varepsilon\big((\delta-2-\varepsilon)n_1+sn_2\big)}\\
    
    \end{bmatrix}_{(\delta-1-\varepsilon)(s+1)\times\varepsilon}.
\end{equation}

The matrix $Y$ is defined in equation \eqref{eq:S.decomposition.Y},
\begin{equation}
\label{eq:S.decomposition.Y}
\begin{split}
    Y := &  diag\big( (\alpha\xi^a)^{i_1}E_{i_1}, (\alpha\xi^a)^{i_2}E_{i_2}, \ldots, (\alpha\xi^a)^{i_\varepsilon}E_{i_\varepsilon} \big), \\
\end{split}    
\end{equation}
where $E_{i_k}:= \sum_{j=0}^{\ell-1}e_{i_kj}v_j$ for all $i_k \in \mathcal{E}$. That is,

\begin{equation}\label{Y}
    Y = 
    \begin{bmatrix}
    (\alpha\xi^a)^{i_1}E_{i_1} & &  & \\
     & (\alpha\xi^a)^{i_2}E_{i_2} & & \\
     & & \ddots & \\
     & & & (\alpha\xi^a)^{i_\varepsilon}E_{i_\varepsilon}\\ 
    \end{bmatrix}_{\varepsilon \times\varepsilon}.
\end{equation}

The matrix $\Tilde{X}$ is defined in equation \eqref{eq:S.decomposition.tilde.X},

\begin{equation}
\label{eq:S.decomposition.tilde.X}
\Tilde{X}:= \big( \xi^{in_1j} \big) _{i \in \{1,2,\ldots,\varepsilon\}}^{j \in [\varepsilon+1]}. 
\end{equation}

That is, 
\begin{equation}\label{TildeX}
    \Tilde{X} = 
    \begin{bmatrix}
    1 & \xi^{n_1} & \xi^{2n_1} & \hdots & \xi^{\varepsilon n_1}\\
    1 & \xi^{2n_1} & \xi^{4n_1} & \hdots & \xi^{2\varepsilon n_1}\\
    \vdots & \vdots & \vdots & \hdots & \vdots\\
    1 & \xi^{\varepsilon n_1} & \xi^{2\varepsilon n_1} & \hdots & \xi^{\varepsilon^2 n_1}\\
    
    \end{bmatrix}_{\varepsilon \times(\varepsilon+1)}.
\end{equation}

%\begin{remark}\label{rem:multi.opn.}
    Consider two matrices $A = [a_{ij}]_{m \times n}$ and $B = [b_{ij}]_{m \times n}$. For a matrix operation, $*$ as defined in \cite{*multi.}, the matrix $A*B$ is defined as follows,
    \begin{equation} \label{eq:star.op}
        A*B =
        \begin{bmatrix}
            a_{11}b_{11} & a_{12}b_{12} & \hdots & a_{1n}b_{1n} \\
            a_{11}b_{21} & a_{12}b_{22} & \hdots & a_{1n}b_{2n} \\
            \vdots & \vdots & \ddots & \vdots \\
            a_{11}b_{m1} & a_{12}b_{m2} & \hdots & a_{1n}b_{mn} \\
            a_{21}b_{11} & a_{22}b_{12} & \hdots & a_{2n}b_{1n} \\
            \vdots & \vdots & \ddots & \vdots \\
            a_{21}b_{m1} & a_{22}b_{m2} & \hdots & a_{2n}b_{mn} \\
            \vdots & \vdots & \ddots & \vdots \\
            \vdots & \vdots & \ddots & \vdots \\
            a_{m1}b_{11} & a_{m2}b_{12} & \hdots & a_{mn}b_{1n} \\
            \vdots & \vdots & \ddots & \vdots \\
            a_{m1}b_{m1} & a_{m2}b_{m2} & \hdots & a_{mn}b_{mn} \\
        \end{bmatrix}.
    \end{equation}
%\end{remark}

Therefore, we can write the decomposition of matrix $X$ as $A*B$, where matrix $A$ defined in equation \eqref{A},

\begin{equation}\label{A}
%\begin{split}
    A := \big( \xi^{ij} \big) _{i \in [s+1]}^{j \in {1,2,\ldots,\varepsilon}} = \begin{bmatrix}
    1 & 1 & \hdots & 1\\
    \xi^{n_2} & \xi^{2n_2} & \hdots & \xi^{\varepsilon n_2}\\
    \vdots & \vdots & \hdots & \vdots\\
    \xi^{sn_2} & \xi^{2sn_2} & \hdots & \xi^{\varepsilon sn_2}\\
    \end{bmatrix}_{(s+1) \times\varepsilon},
%\end{split}
\end{equation}
and matrix $B$ defined in equation \eqref{B},
\begin{equation}\label{B}
%\begin{split}
    B := X^{\left<0\right>} =\begin{bmatrix}
    1 & 1 & \hdots & 1\\
    \xi^{n_1} & \xi^{2n_1} & \hdots & \xi^{\varepsilon n_1}\\
    \vdots & \vdots & \hdots & \vdots\\
    \xi^{(\delta-2-\varepsilon)n_1} & \xi^{2(\delta-2-\varepsilon)n_1} & \hdots & \xi^{\varepsilon (\delta-2-\varepsilon)n_1}\\
    \end{bmatrix}_{(\delta-1-\varepsilon)\times\varepsilon},
%\end{split}
\end{equation}

We now present a syndrome-based decoding procedure for quasi-twisted codes. This algorithm can correct up to $\varepsilon = \frac{d^*-1}{2}$ errors, where $d^* = min(\delta+s, d_\mathbb{C})$ as defined in theorem \ref{HT-like Bound-Th}. Table \ref{tab:dec.algo} contains a complete description of our decoding algorithm for an $[m \ell, k, d]_q$ $(\lambda,\ell)$-quasi-twisted code $\mathcal{C}$.

\begin{table}[h!]
    \centering
    \begin{tabular}{|l|}
    \hline
    Decoding algorithm for a $(\lambda,\ell)$-QT code\\ 
    \hline
    \hfill \\
    \textbf{Input:}\\ $\lambda,m,\ell,k,q,r \longleftarrow$ parameters of the quasi-twisted code $\mathcal{C}$\\
    $r(X) = \big(r_0(X), r_1(X), \ldots, r_{\ell-1}(X) \big) \in 
    \mathbb{F}_q[X]^\ell \longleftarrow$ received word\\
    $a \geq 0, \delta \geq 2, \{s, n_1,n_2\} \in \mathbb{Z}^+$ with $gcd(m,n_1)=1, gcd(m,n_2)<\delta$\\
    $\beta = \alpha\xi^{a+in_1+jn_2},$ $\forall i \in [\delta-1], j \in [s+1] \longleftarrow$ eigenvalues\\
    $\textbf{v} = (v_0, v_1, \ldots, v_{\ell-1}) \in \mathbb{F}_{q^r}^\ell \longleftarrow$ eigenvector\\
    \hfill \\
    
    \textbf{Output:}\\
    $c(X) = (c_0(X), c_1(X), \ldots, c_{\ell-1}(X)) \longleftarrow$ estimated codeword\\
    or\\
    DECODING FAILURE\\
    \hfill \\

    \textbf{Algorithm:}\\
    \tabitem Compute $S_t(X), \forall t \in [s+1]$ as in equation \eqref{S(X)}\\
    \tabitem Solve the Key Equations jointly by applying the generalized \\
    Berlekamp-Massey algorithm on $(X^{\delta-1}, S_0(X), S_1(X), \ldots, S_s(X))$ \\(refer remark \ref{Key-Eqns.})\\
    \tabitem Find all $i_k$ such that $\Lambda(\xi^{-i_k
    n_1})=0$ $\implies$ $\mathcal{E} = \{i_1,i_2,\ldots,i_\varepsilon\}$\\
    \tabitem If $\varepsilon < deg(\Lambda(X)):$\\
    \hspace{10mm} DECODING FAILURE $\longleftarrow$ \textbf{Output}\\
    \hspace{4mm} else:\\
    \hspace{10mm} Compute $E_{i_1},E_{i_2},\ldots,E_{i_\varepsilon} \in \mathbb{F}_{q^r}$ from one of the error evaluator \\
    \hspace{10mm} polynomials, $\Omega_t(X), t \in [s+1]$ \\
    \hspace{10mm} Compute $e_{{i_k},0}, e_{{i_k},1}, \ldots, e_{{i_k},{\ell-1}} \in \mathbb{F}_q$ s.t. $E_{i_k}:= \sum\limits_{j=0}^{\ell-1}e_{i_kj}v_j, \hspace{2mm} \forall i_k \in \mathcal{E}$\\
    \hspace{10mm} Compute $e_j(X) = \sum\limits_{i \in \mathcal{E}_j} e_{i,j}X^i, \hspace{2mm} \forall j \in [l]$\\
    \hspace{10mm} Compute $c_j(X) = r_j(X) - e_j(X), \hspace{2mm} \forall j \in [l]$\\
    \hspace{10mm} $c(X) = (c_0(X), c_1(X), \ldots, c_{\ell-1}(X)) \longleftarrow$ \textbf{Output}\\
    \hfill\\
    \hline
    \end{tabular}
    \caption{We present a step-by-step syndrome-based decoding algorithm for quasi-twisted codes that can correct up to $\frac{d^*-1}{2}$ errors. The first step in the algorithm computes $s+1$ syndrome polynomials. In the following step, the algorithm jointly solves the linear system of $s+1$ 'Key Equations' using a generalized Berlekamp-Massey algorithm. The third step involves solving for the roots of the error-locator polynomial $\Lambda(X)$ by using Chien search algorithm, which gives the positions of the burst errors, $i_k \in \mathcal{E}$. If the total number of errors, $\varepsilon$ is less than the degree of $\Lambda(X)$, then the algorithm fails. Otherwise, the algorithm computes the corresponding error values, $E_{i_k}$ using one of the error-evaluator polynomials $\Omega(X).$ Now, each of these $E_{i_k} \in \mathbb{F}_{q^r}$ corresponds to an $e_{i_kj} \in \mathbb{F}_q$, which gives the coefficients of the error polynomial $e(X)$. Finally, the algorithm reconstructs the codeword polynomial $c(X)$ from $e(X)$, which is the output of the algorithm.}
    \label{tab:dec.algo}
\end{table}

We now analyze the performance of our decoding algorithm. We begin by proving a condition that ensures that the decoding algorithm can correct up to $\varepsilon$ errors in theorem \ref{thm:rank.syndrome.matrix}. 

\begin{theorem}\label{thm:rank.syndrome.matrix}
    Consider an $[m\ell, k, d]_q$ $(\lambda,\ell)$-quasi-twisted code $\mathcal{C}$ such that condition \eqref{Th-eqn} together with the assumptions of theorem \ref{HT-like Bound-Th} hold. 
    Let  $\mathcal{E} = \{i_1, i_2, \ldots, i_\varepsilon\}$ be the set of error locations such that $\varepsilon$ is the maximum number of errors.
    Consider the $(s+1)$ syndrome polynomials $S_0(X), S_1(X),\ldots,S_s(X)$ as defined in equation \eqref{S(X)}. 
    Then, rank of the syndrome matrix defined in equation \eqref{S} is $rank(S) = \varepsilon.$
\end{theorem}

\begin{proof}
Let the syndrome matrix $S$ and the matrices $X, Y, \Tilde{X}$ in its decomposition, $S=XY\Tilde{X}$ be as defined in equations \eqref{S} to \eqref{X}. 
The matrix $\Tilde{X}$ is a Vandermonde matrix. Since $gcd(m,n_1)=1$, we have $rank(\Tilde{X}) = min(\varepsilon,\varepsilon+1) = \varepsilon$. The diagonal matrix, $Y$ is non-singular. Therefore, we have $rank(S) = rank(X)$.

Consider the decomposition $X = A*B$, where $A$ and $B$ are as defined in equations \eqref{A} and \eqref{B}. Using \cite[Section VI.A]{Gen.BM-Algo}, we see that if $rank(A) + rank(B) > \varepsilon$, then $rank(X) = \varepsilon$.  
Notice that both matrices $A$ and $B$ are Vandermonde. Given $gcd(m,n_1)=1$, we have $rank(A) = min((s+1),\varepsilon)$ and $rank(B) = min((\delta-1-\varepsilon),\varepsilon).$

Assume  $\delta-1>s$ (as shown in \cite{Gen.BM-Algo} for HT-bound). Recall that $\varepsilon \leq \frac{d^*-1}{2},$ where $d^* = min(\delta + s, d_\mathbb{C}).$ We get
\[\varepsilon \leq \frac{d^*-1}{2} = \frac{min(\delta + s, d_\mathbb{C})-1}{2} \leq \frac{\delta + s -1}{2} < \delta-1.\]
We can simplify $rank(A)+rank(B)$ in one of the four following ways.
\begin{itemize}[noitemsep]
    \item[i.] $(s+1) + (\delta-1-\varepsilon) = (\delta+s-1) +1-\varepsilon > 2\varepsilon + 1 -\varepsilon = \varepsilon +1 > \varepsilon,$\\
    \item[ii.] $(s+1) + \varepsilon > \varepsilon,$ \\
    \item[iii.] $\varepsilon + (\delta-1-\varepsilon) = \delta-1 > \varepsilon,$ \\
    \item[iv.] $\varepsilon + \varepsilon > \varepsilon$
\end{itemize}
Note that $rank(A) + rank(B) > \varepsilon$ in all four cases. Thus, we get $rank(X) = \varepsilon$, and hence $rank(S) = rank(X) = \varepsilon.$ 

\end{proof}

\begin{remark}
   The condition $rank(S) = \varepsilon$ guarantees that the error locator polynomial in equation \eqref{Lambda_poly} has unique roots $\xi^{-in_1}, \forall i \in \mathcal{E}$. Therefore, provided the number of errors is at most $\varepsilon$, our decoding algorithm as described in table \ref{tab:dec.algo} outputs a codeword of the chosen quasi-twisted code. 
\end{remark}

Now we analyze the time complexity of the algorithm to gauge its efficiency. The result is summarized in theorem \ref{thm:dec.alg.complexity}. 

\begin{theorem}[\bf Complexity analysis]\label{thm:dec.alg.complexity}
     Consider an $[n = m\ell, k, d]_q$ $(\lambda,\ell)$-quasi-twisted code $\mathcal{C}$. The syndrome-based decoding algorithm given in table \ref{tab:dec.algo} has $\mathcal{O}(n^2)$ time complexity.
\end{theorem}

\proof

The first line the algorithm computes syndrome polynomials. That is, it computes $\ell$ dot products and evaluates $\delta-1$ coefficients for $\ell$ polynomials each of length at most $m$. Therefore, the complexity of this step is $\mathcal{O}(n) = \mathcal{O}(m\ell)$. Given that we employ the generalized Berlekamp-Massey algorithm in step two, its complexity is $\mathcal{O}(\sqrt{n})$. In the next step, the algorithm computes roots of the error locator polynomial using Chien search. The complexity of step three is $\mathcal{O}(\varepsilon)$, that is $\mathcal{O}(1)$ with respect to $n$. Furthermore, computing the error polynomial and codeword polynomial has $\mathcal{O}(1)$ complexity with respect to $n$.

It is clear that the highest complexity corresponds to the second step of the algorithm (application of the generalized Berlekamp-Massey algorithm). Therefore, the overall complexity of our syndrome-based decoding algorithm is $\mathcal{O}(n^2).$ 

\qed

We end this section with an example of how our decoding algorithm works for the $[20,10,4]_3$ $(2,2)$-QT code.

\begin{example}
    Consider the $[10\cdot 2, 10, 4]_3$ $2$-quasi-twisted code $\mathcal{C}$ with its $2 \times 2$ generator matrix expressed in the reduced Gröbner basis form as follows.
    \begin{equation}
        \Tilde{G}(X) = 
        \begin{bmatrix}
            1 & 2X^9 + 2X^7 + 2X^6 + X^5 + 2X^3 + X^2 + 1\\
            0 & X^{10} + 1
        \end{bmatrix}
    \end{equation}
    Let $\mathbb{F}_3(a) := \mathbb{F}_{3^4}$ be the splitting field of $X^{10}+1.$ Notice that, we have $(m,\ell, \lambda) = (10,2,2).$ Let $\alpha$ and $\xi$ be the fixed $m^{th}$ roots of $\lambda$ and unity, respectively. We have $\alpha = a^4$ and $\xi = a^8.$ For subset, $\Omega = \{a^{52}, a^{60}, a^{68}\}$ of eigenvalues, we have $\beta = \alpha\xi^i$ for all $i \in D =\{6,7,8\}$. Consider the eigenvector $\mathbf{v} = (1 \hspace{3mm} a^{50}) \in \mathcal{V}$, where $\mathcal{V}$ is the intersection of eigenspaces corresponding to the eigenvalues in $\Omega$.

    In the case that an all-zero codeword is transmitted, and the received polynomial is $r(X) = (0, \hspace{2mm} X^8)$, we have
    \begin{equation}
        \begin{split}
            r_0(X) = & e_0(X) = 0 \text{ (the zero polynomial)} \\
            r_1(X) = & e_1(X) = X^8
        \end{split}
    \end{equation}
    Therefore, $\varepsilon = 1$. By solving the following system of equations, we can find the coefficient $\Lambda_1$ in the error-locator polynomial (as defined in equation \eqref{Lambda_poly}):
    \begin{equation}
        \begin{bmatrix}
            S_0^{\left<0\right>} \\
            S_1^{\left<0\right>} 
        \end{bmatrix}
        \begin{bmatrix}
            \Lambda_\varepsilon
        \end{bmatrix}
        =
        \begin{bmatrix}
            S_1^{\left<0\right>} \\
            S_2^{\left<0\right>} 
        \end{bmatrix}
    \end{equation}
    That is,
    \begin{equation}
        \begin{bmatrix}
            a^{66} \\
            a^{50}
        \end{bmatrix}
        \begin{bmatrix}
            \Lambda_1
        \end{bmatrix}
        =
        \begin{bmatrix}
            a^{50}\\
            a^{34}
        \end{bmatrix}
    \end{equation}
    Therefore, we get $\Lambda_1 = a^{64} = \xi^8$, which means the error-locator polynomial is given by, $\Lambda(X) = 1-X\xi^8.$ Upon error-evaluation, we get $E_8 = a^{66}$ in $\mathbb{F}_{81}$ and the error value is $e_{8,1} = 1$ in $\mathbb{F}_3.$ Therefore, our decoding algorithm corrects the received word generating the codeword polynomial corresponding to an all-zero codeword.
    
\end{example}

\section{Niederreiter-like code-based cryptosystem}\label{sec:Niederreiter cryptosystem}

In this section, we propose a Niederreiter-like code-based cryptosystem based on quasi-twisted codes and analyze its security in certain situations. Let $\ell,m, \lambda$ be positive integers and $m$ be a prime power of any prime other than the characteristic of $\mathbb{F}_q$, and integer $\ell \leq \mathcal{O}(poly)(m)$. Consider an $[n=m\ell, k=(\ell-1)m]_q$ $(\lambda,\ell)$-quasi-twisted code $\mathcal{C}$. Suppose $c = (c_0, c_1, c_2, \ldots, c_{n-1}) \in \mathcal{C}$ is a codeword of $\mathcal{C}$. Then, by definition \ref{def.QT}, we know that 
%\[ \text{for } c = \{a_0, a_1, a_2, \ldots, a_{n-1}\} \in \mathcal{C},\]
$(\lambda c_{n-\ell}, \lambda c_{n-\ell+1}, \ldots, \lambda c_{n-1}, c_0, \ldots, c_{n-\ell-1}) \in \mathcal{C}$ is also a codeword of $\mathcal{C}$. 

\begin{definition}(\textbf{Twistulant matrix})
    %An $m \times m$ twistulant matrix is defined as:
    Consider $c = (c_0, c_1, \hdots, c_{m-1}) \in \mathbb{F}_q^m.$ An $m \times m$ matrix is called a (right) twistulant matrix if its rows are composed of the (right) $\lambda$-constashifts of $c$. Any matrix $G$ is a twistulant matrix if $c$ forms the first row of $G$ and every other row of $G$ can be obtained by a $\lambda$-constashift of the row above it. That is, matrix $G$ is of the following form.
    \begin{equation}
        G =
        \begin{bmatrix}
        c_0 & c_1 & c_2 & \hdots & c_{m-1}\\
        \lambda c_{m-1} & c_0 & c_1 & \hdots & c_{m-2}\\
        \lambda c_{m-2} & \lambda c_{m-1} & c_0 & \hdots & c_{m-3}\\
        \vdots & \vdots & \vdots & \ddots & \vdots \\
        \lambda c_1 & \lambda c_2 & \lambda c_3 & \hdots & c_0
        \end{bmatrix}
    \end{equation}
\end{definition}

A twistulant matrix can be represented by its first row. Notice that this row can be mapped to the polynomial $c(X) = c_0 + c_1X + c_2X^2 + \hdots + c_{m-1}X^{m-1} \in \mathbb{F}_q[X]/\left< X^m - \lambda\right>$. This polynomial $c(X)$ is called the \textit{defining polynomial} of the twistulant matrix $G.$ When $\lambda = 1,$ a twistulant matrix is a circulant matrix.
 
It is proved in \cite{r-gen.QT} that a generator matrix of a $(\lambda,\ell)$-QT code can be expressed in the form of blocks of twistulant matrices. Each row of such a matrix is conventionally termed a \textit{generator}. Therefore, the generator matrix of a 1-generator QT-code can be expressed as,
\begin{equation}\label{1-gen.matrix}
    G = \Big[G_0| G_1| \hdots| G_{\ell-1}\Big],
\end{equation} 
where each $G_i$ is an $m \times m$ twistulant matrix.

Note that the dual code, $\mathcal{C}^\perp$ of the $[m\ell, (\ell-1)m]$ $(\lambda,\ell)$ 1-generator QT code, an $[m\ell,m]$ $(\lambda^{-1},\ell)$-QT code \cite{QT-codes}, is also a 1-generator code.

Thus, the generator matrix of $\mathcal{C}^\perp$ will be of the form given in equation \eqref{1-gen.matrix} with each $G_i$ as defined in equation \eqref{eq:G_i}.
\begin{equation}\label{eq:G_i}
    G_i =
    \begin{bmatrix}
    c_0 & c_1 & c_2 & \hdots & c_{m-1}\\
        \lambda^{-1} c_{m-1} & c_0 & c_1 & \hdots & c_{m-2}\\
        \lambda^{-1} c_{m-2} & \lambda^{-1} c_{m-1} & c_0 & \hdots & c_{m-3}\\
        \vdots & \vdots & \vdots & \ddots & \vdots \\
        \lambda^{-1} c_1 & \lambda^{-1} c_2 & \lambda^{-1} c_3 & \hdots & c_0
    \end{bmatrix}
\end{equation}
V. Bhargava et. al (\cite{G.std.form}) describe how a generator matrix of a rate $1/m$ code can be expressed in the standard form. Since $\mathcal{C}^\perp$ is a rate $1/m$ code, if one of the $G_i's$ is invertible (say, $G_0$), we can rewrite $G$ in the standard form, $G^*$.
\begin{equation}\label{H matrix}
    G^{*} = \Big[\mathcal{I}| G_1^{*}| G_2^*| \hdots| G_{\ell-1}^{*}\Big],
\end{equation}
where $\mathcal{I}$ is the $m \times m$ identity matrix and $G_i^{*} = G_i(G_0^{-1}) = (G_0^{-1})G_i,$ %\RK{write in inverse format. Right multiplication or left multiplication?} 
$\forall$ $1 \leq i \leq \ell-1$ (invertible twistulant matrices form an abelian group). The set of all right twistulant matrices is closed under multiplication (\cite{twistulant1}. Moreover, for an invertible $\lambda$-twistulant matrix, its inverse is also a $\lambda$-twistulant matrix (\cite{twistulant2}). Using these two facts, we can clearly see that $G_i^*$ ($\forall$ $1 \leq i \leq \ell-1$) is also a $\lambda$-twistulant matrix. Since $G^{*}$ forms a generator matrix of $\mathcal{C}^\perp$, 
we can consider $H = G^{*}$ to be a parity check matrix of $\mathcal{C}$, the quasi-twisted code of our interest.

\subsection{Proposed cryptosystem}\label{subsec:cryptosystem.desc}

In this section, we describe the proposed Niederreiter-like cryptosystem based on quasi-twisted codes.

Our cryptosystem differs from the Niederreiter cryptosystem in the key generation step. We choose an $[n=m\ell, k=(\ell-1)m]_q$ $(\lambda,\ell)$-quasi-twisted code $\mathcal{C}$, where $\ell,m, \lambda$ are positive integers such that $m$ is a prime power for some prime other than the characteristic of $\mathbb{F}_q$. Let $\ell$ be bounded above by a polynomial in $m$. The parity check matrix, $H$ of this code $\mathcal{C}$ is chosen such that $H = G^*$ as defined in equation \eqref{H matrix}. The encryption and decryption steps are same as described in figure \ref{fig:Niederreiter}.

\subsection{Classical security}

We now analyze the classical security of our cryptosystem. The equivalence of McEliece and Niederreiter cryptosystems from a security perspective is proved in \cite{McE=Nied}. Two important classical attacks against McEliece or Niederreiter cryptosystems are considered in this section.
\begin{enumerate}
    \item Information Set Decoding (ISD) attacks:
    ISD attack is one of the best-known generic attacks against code-based cryptosystems, with least complexity. The attack relies on a random selection of linearly independent columns of the public key, $H'$, termed as \textit{'information sets'} under the assumption that these do not contain errors. The computational complexity of ISD attacks can be used to determine the minimum public key size that is required to attain a given security level of the cryptosystem.
    
    These can be induced by two strategies; Lee and Brickell's method and Stern's algorithm. The former uses a probabilistic approach to find information sets containing fewer errors. The latter employs a meet-in-the-middle technique to reduce the search space by splitting error vectors into partial sums. Furthermore, these attacks have a minimum work factor, which represents the expected number of elementary operations required to successfully decode a random linear code (\cite{work.factor}). Therefore, the minimum work factor can be considered as the security level of these cryptosystems. 
    Each iteration of an ISD algorithm requires inverting submatrices corresponding to chosen information sets. The work factor (complexity) for the matrix inversion is $O(k^\gamma)$, where $ 2 < \gamma < 3,$ and $k$ is the size of the information set. 

Now we analyze the minimum work factor for our proposed cryptosystem based on an  $[m\ell, (\ell-1)m]$ $(\lambda,\ell)$-QT code. Let the error-correcting capacity of this code be $\varepsilon$. To compute work factor for matrix inversion, we use the elementary algorithm in \cite{alpha-work.factor}, which is $\alpha (n-k)^3 = \alpha m^3$ for a small constant $\alpha.$ We replace $k$ in the original equation with $n-k$ since we use a variant of the Niederreiter cryptosystem that uses a parity check matrix in place of a generator matrix. 

 The probability that there is no error in randomly selected $n-k =m$ columns (information sets) from n columns (with $\varepsilon$ errors) of $H'$ is $\binom{n-\varepsilon}{n-k}/\binom{n}{n-k} = \binom{m\ell-\varepsilon}{m}/\binom{m\ell}{m}$ (\cite{Lee.Brickell}). Thus, the total work factor for one iteration of an ISD attack is given by,
\begin{equation}
    W = \frac{\alpha m^3 \binom{m\ell}{m}}{\binom{m\ell-\varepsilon}{m}}.
\end{equation}
Optimization parameters for our analysis are chosen based on Lee and Brickwell's framework (\cite{Lee.Brickell}). Therefore, the minimum work factor is given by,
\begin{equation}
    W_{min} = T_2(\alpha m^3 + N_2 \beta m) = T_2(m^3 + N_2 m),
\end{equation}
\[\text{where, }T_2 = \frac{1}{\sum\limits_{i=0}^{2} Q_i }, 
Q_i = \frac{{\binom{\varepsilon}{i}}\binom{m\ell-2}{m-i}}{\binom{n}{m}}
\text{ and } N_2 = \sum_{i=0}^{2}\binom{m}{i}.\]

Setting $\alpha=\beta=1$
, we get,
\begin{equation}
    W_{min} = T_2\big(m^3 + mN_2 \big)
\end{equation}

We can now use this minimum work factor to determine minimum size of the public key required to compute security level of our cryptosystem. Due to lack of literature on exact constructions of quasi-twisted codes in the form required by our construction, we are unable to present a numerical analysis of this parameter for our cryptosystem.

%This is because the minimum work factor is indicative of the lower bound on computational effort required by an adversary. 

\item Attack on dual code: An attack on the dual code is another class of significant attacks on code-based cryptosystems. These attacks exploit the structural weaknesses the underlying code. A cryptosystem is prone to such an attack when the parity check matrix is sparse, resulting in a dual code with low-weight codewords. When there are many low-weight codewords in the dual code, an attacker can exploit patterns to recover information about the underlying code structure, and in turn the secret key. A dual code attack proceeds by analyzing the weight distribution of the dual code so as to understand linear dependencies in the parity check matrix. This poses a security threat to the cryptosystem. Such an attack can be prevented by avoiding sparse parity check matrices. The choice of parameters for our cryptosystem ensures that the chosen parity check matrix is not sparse. Therefore, our cryptosystem is secure against attacks on the dual code. 

\end{enumerate}

\subsection{Quantum security}
We now analyze security of our cryptosystem against some quantum attacks. Quantum Fourier Sampling (QFS) plays a central role in many quantum algorithms, including Shor’s algorithm. A detailed background can be found in \cite{QFS.attacks}. We show that our cryptosystem can withstand attacks based on Quantum Fourier Sampling (QFS).

\subsubsection{Hidden subgroup problem}
One way of attacking the Niederreiter cryptosystem is a scrambler pemutation attack, which can be reduced into an instance of the hidden subgroup problem. QFS can be used to solve the hidden subgroup problem. We start with defining these problems.

\begin{definition}{(\bf Scrambler permutation attack)}
    This attack relies on finding the scrambler permutation pair, that is a matrices $S$ and $P$, assuming that $H$ and $H'$ are known. Any $S', P'$ satisfying $H' = S'HP'$ is enough to make the attack successful.
\end{definition}

\begin{definition}{(\bf Hidden shift problem)}
    Consider a finite group, $G$ and a finite set, $\Sigma$. Define two functions, $f_0, f_1:G \longrightarrow \Sigma$. The problem is to find a constant $g \in G$ such that $f_1(x) = f_0(gx)$ for all $x \in G$, provided such a constant (termed as `shift' from $f_0$ to $f_1$) exists. 
    There might exist more than one such shifts; finding any one of these is enough.
\end{definition}

\begin{definition}{(\bf Hidden subgroup problem)}
   Let $f$ be a function on a group $G$ such that $f(x_1) = f(x_2)$ if and only if %\iff 
    $x_1H = x_2H$, for some unknown subgroup $H< G$. The problem is to find a set of generators for the subgroup $H$.
\end{definition}

The scrambler permutation problem reduces to the hidden shift problem when the group is taken to be $G = GL_{n-k}(\mathbb{F}_q)\times S_n$ and the functions are defined on this group such that for all $(S, P) \in GL_{n-k}(\mathbb{F}_q)\times S_n,$
\begin{equation}
    f_0(S,P) = S^{-1}HP, \hspace{5mm} f_1(S,P) = S^{-1}H'P.
\end{equation}
Here, the permutation matrix $P$ is mapped to the corresponding permutation in $S_n$. Note that $H' = SHP$ if nd only if $(S^{-1},P)$ is a shift from $f_0$ to $f_1.$ That is, for all $(S,P) \in GL_{n-k}(\mathbb{F}_q)\times S_n$,
\begin{equation}
    f_0((S^{-1},P)(S,P)) = S^{-1}(SHP)P = S^{-1}H'P = f_1(S,P).
\end{equation}

Let $G \wr \mathbb{Z}_2$  represent the wreath product and $G^2 \rtimes \mathbb{Z}_2$ represent the semi-direct product for a group $G$ and set $\mathbb{Z}_2 := \{0,1\}$. The hidden shift problem on a group $G$ reduces to the hidden subgroup problem on $G \wr \mathbb{Z}_2 = G^2 \rtimes \mathbb{Z}_2.$ Consider two functions $f_0$ and $f_1$ defined on $G$ and define a function $f: G \wr \mathbb{Z}_2 \longrightarrow \Sigma\times\Sigma.$ For $(x_1,x_2) \in G^2$ and $a \in \mathbb{Z}_2,$
\begin{equation}
    f\big((x_1,x_2),a\big) :=
    \begin{cases}
        \big(f_0(x_1), f_1(x_2)\big) & \text{if } a =0 \\
        \big(f_1(x_2), f_0(x_1)\big) & \text{if } a = 1 \\
    \end{cases}
\end{equation}

The hidden shift problem for $G = GL_{n-k}(\mathbb{F}_q)\times S_n$ and shift from $f_0$ to $f_1$ to be $s$, reduces to the hidden subgroup problem on $G^2 \rtimes \mathbb{Z}_2 = (GL_{n-k}(\mathbb{F}_q)\times S_n)^2 \rtimes \mathbb{Z}_2.$ If $H_0 := G|_{f_0},$ then the hidden subgroup is given by
\begin{equation}
    K := G \wr \mathbb{Z}_2|_f = \Bigg(\bigg(\Big(H_0, s^{-1}H_0s\Big),0\bigg) \bigcup \bigg(\Big(H_0s, s^{-1}H_0\Big),1\bigg)\Bigg)
\end{equation}

Finding this hidden subgroup $K = G \wr \mathbb{Z}_2|_f$ enables in finding a shift from $f_0$ to $f_1.$ That is, if $\big((g_1,g_2),1\big) \in K,$ then there exists $g_1 \in H_0s$ such that $g_1$ is a shift from $f_0$ to $f_1.$ We can verify by computing $s = (S^{-1},P).$ 
Consider
\begin{equation}
    H_0 := G|_{f_0} = \{(S,P) \in GL_{n-k}(\mathbb{F}_q)\times S_n : S^{-1}HP = H \}
\end{equation}
Let $(\mathscr{S}, \mathscr{P}) \in H_0.$ Then, by definition, we have 
\[\mathscr{S}^{-1}H\mathscr{P} = H\].
So, 
\begin{equation}
    (\mathscr{S}, \mathscr{P})(S^{-1}P) = (\mathscr{S}S^{-1}, \mathscr{P}P) \in H_0s
\end{equation}
Therefore,
\begin{equation}
    f_0 \Big( \big( \mathscr{S}S^{-1}, \mathscr{P}P\big)\big(S,P\big)\Big) = S^{-1}S\mathscr{S}^{-1}H\mathscr{P}PP = S^{-1}SHPP = S^{-1}H'P = f_1(S,P)
\end{equation}
This proves that $(\mathscr{S}S^{-1}, \mathscr{P}P) \in H_0S$ is a shift.

Solving the hidden subgroup problem on $(GL_{n-k}(\mathbb{F}_q) \times S_n)^2 \rtimes \mathbb{Z}_2$ for hidden subgroup, $K$ is same as solving the hidden shift problem over $GL_{n-k}(\mathbb{F}_q)\times S_n$. This, in turn, guarantees a solution to the scrambler permutation problem, which is a basis for the scrambler permutation attack. Therefore, we focus on difficulty of solving the hidden subgroup problem on $(GL_{n-k}(\mathbb{F}_q) \times S_n)^2 \rtimes \mathbb{Z}_2$ to prove resistance of our cryptosystem to a scrambler permutation attack.

\subsubsection{Indistinguishability by QFS}
 
Let $G$ be a finite group. Consider QFS over $G$ in basis, $\{B_\rho\}$. 
Two subgroups, $H_1$ and $H_2$, are said to be indistinguishable if their probability distributions, $P_{H_1}$ and $P_{H_2}$, have their total variations really close to one another. In this section, we describe a sufficient condition for indistinguishability of  subgroup, $K$ using some of the results in \cite{QFS.attacks}. We check indistinguishability of a hidden subgroup $H < G$ from its conjugate subgroups $gHg^{-1}$ or the trivial subgroup $\left<e\right>$, where $e$ is the identity element in $G$. Weak Fourier sampling cannot distinguish between conjugate subgroups since $P_{gHg^{-1}}$ does not depend on $g$, which implies that conjugate subgroups have the same probability distribution.
 
 Let $\hat{G}$ be the set of irreducible unitary representations of $G$ and $\rho \in \hat{G}$ be an irreducible representation given by weak Fourier sampling. For some non-trivial subgroup $H$, we show that strong Fourier sampling cannot efficiently distinguish between $H$ and its conjugates or from the trivial subgroup. Using strong Fourier sampling, we show that the probability distribution of the trivial hidden subgroup, $P_{\left< e \right>} (\cdot|\rho)$ is same as the uniform distribution $U_{B_\rho}$ on the basis $B_\rho$. Thus, it is enough to show that for a random $g \in G$, $P_{gHg^{-1}}(\cdot|\rho)$ is close to $U_{B_\rho}$ in total variation.

 We now restate the formal definition of distinguishability of a subgroup using strong QFS as given in \cite[Defn. 5]{QFS.attacks}.

\begin{definition}{(\bf Distinguishability of a subgroup using strong QFS)}\label{D_H defn}
    Let $\mathscr{D}_H$ denote the distinguishability of a subgroup $H < G$ using strong Fourier sampling over $G$. Then, $\mathscr{D}_H$ is defined to be the expectation of the squared $\ell_1$-distance between $P_{gHg^{-1}}(\cdot|\rho)$ and $U_{B_\rho}.$ That is,
    \begin{equation}
        \mathscr{D}_H:= \mathbb{E}_{\rho,g}\Big[\left\|P_{gHg^{-1}}(\cdot|\rho) - U_{B_\rho}\right\|_1^2\Big]
    \end{equation}
    for $\rho \in \hat{G}$ and a random $g \in G.$ The subgroup $H$ is said to be indistinguishable if,
    \begin{equation} \label{indist.cond.}
        \mathscr{D}_H \leq \log^{-\omega(1)}|G|
    \end{equation}
\end{definition}
\vspace{-2mm}
Using Markov's inequality, if strong Fourier sampling cannot distinguish the subgroup $H$, then for all constant $c>0,$
\begin{equation}
    \left\|P_{gHg^{-1}}(\cdot|\rho) - U_{B_\rho}\right\|_{t.v.} < \log^{-c}|G|
\end{equation}
with a minimum probability of $1-\log^{-c}|G|$ in both $g$ and $\rho.$

\subsubsection{Security of our cryptosystem against QFS based attacks }\label{subsec:Nied-like.crypto}

Recall that we consider a Niederreiter-like cryptosystem where the underlying code is an $[n=m\ell, k=(\ell-1)m]$ $(\lambda,\ell)$-quasi-twisted code $\mathcal{C}$ with parity check matrix, $ H = \Big[\mathcal{I}|     G_1^{*}| G_2^{*}| \hdots| G_{\ell-1}^{*}\Big]_{m \times m\ell}$ as defined in equation \eqref{H matrix}. Here, $\mathcal{I}$ is the $m \times m$ identity matrix and for all $1 \leq i \leq \ell-1$, $G_i^{*}$ is a $\lambda^{-1}$-twistulant matrix with entries in $\mathbb{F}_q$. That is, 
%$G_i^{*} = G_i/G_0,$ $\forall$ $1 \leq i \leq \ell-1$, with
\begin{equation}
    G_i^* =
    \begin{bmatrix}
    c_0 & c_1 & c_2 & \hdots & c_{m-1}\\
    \lambda^{-1} c_{m-1} & c_0 & c_1 & \hdots & c_{m-2}\\
    \lambda^{-1} c_{m-2} & \lambda^{-1} c_{m-1} & c_0 & \hdots & c_{m-3}\\
    \vdots & \vdots & \vdots & \ddots & \vdots \\
    \lambda^{-1} c_1 & \lambda^{-1} c_2 & \lambda^{-1} c_3 & \hdots & c_0
    \end{bmatrix}_{m \times m}
\end{equation}
Note that $n-k = m$ for the $[m\ell,(\ell-1)m]_q$ code under consideration. 

\begin{remark}\label{rem.twistulant}
    We impose following conditions on the parity check matrix $ H = \Big[\mathcal{I}|     G_1^{*}| G_2^{*}| \hdots| G_{\ell-1}^{*}\Big]$.
    \begin{enumerate}[topsep = 0pt]
        \item No two twistulant matrices $G_i^*$ and $G_j^*,$ for $i \neq j, 1 \leq i,j \leq \ell-1,$ can have the same defining polynomial, $c(X) = c_0 + c_1X + c_2X^2 + \hdots + c_{m-1}X^{m-1} \in \mathbb{F}_q[X]/\left< X^m - \lambda\right>$. \\
        \item If $\lambda = 1$ (or equivalently, if $q =2$), at least one of the coefficients of the defining polynomial, $c(X),$ should be distinct from the rest of the coefficients of $c(X)$. In other words, there exists  $i,$ $1 \leq i \leq m-1$ such that $c_i \neq c_j$ for all $j \neq i,$ $1 \leq j \leq m-1.$
        %If $\lambda = 1$ or $q =2,$ then the non-zero co-efficients in the defining polynomial of $G_i^*,$ for all $1 \leq i \leq \ell-1,$ must be distinct.
    \end{enumerate}
\end{remark}

Recall that the scrambler permutation attack on a Niederreiter cryptosystem is same as solving the hidden subgroup problem on $(GL_{m}(\mathbb{F}_q)\times S_n)^2 \rtimes \mathbb{Z}_2$ with hidden subgroup, $K$. For a hidden element $s \in GL_{m}(\mathbb{F}_q)\times S_n$, we have
\begin{equation}\label{subgroupK}
    K = G \wr \mathbb{Z}_2|_f = \Bigg(\bigg(\Big(H_0, s^{-1}H_0s\Big),0\bigg) \bigcup \bigg(\Big(H_0s, s^{-1}H_0\Big),1\bigg)\Bigg),
\end{equation}
where
\begin{equation}
    H_0 := G|_{f_0} = \{(S,P) \in GL_{m}(\mathbb{F}_q)\times S_n : S^{-1}HP = H\}.
\end{equation}
%Let us consider $q=p$ for some prime number $p.$ 
The \textbf{automorphism group},$(Aut(H))$ of a linear code generated by $H$ as the projection of $H_0$ onto $S_n.$ That is,
\begin{equation}
    Aut(H) = \{P \in S_n: S^{-1}HP = H \text{ for some }S \in GL_{m}(\mathbb{F}_p)\}.
\end{equation}
Therefore, for all $(S,P) \in H_0$, there exists $P \in Aut(H)$.

\begin{definition}{(\bf Minimal degree)}\label{min.deg.}
    The minimal degree of a permutation group is the number of points that are not fixed by a non-trivial element of the group.
\end{definition}

For a McEliece cryptosystem, H. Dinh et al. (\cite[Thm. 4]{QFS.attacks}) give a bound on $\mathscr{D}_K$ in terms of the minimal degree of $Aut(M),$ where $M$ is the $k \times n $ generator matrix of the underlying code. We can rewrite \cite[Thm. 4]{QFS.attacks} for a Niederreiter cryptosystems in terms of the minimal degree of $Aut(H)$, where $H$ is an $(n-k) \times n$ parity check matrix of the underlying code. 
\begin{theorem}\label{thm.indist.}
    Assume $q^{{(n-k)}^2} \leq n^{an}$ for some constant $0 < a < 1/4$. Let $d$ be the minimal degree of an automorphism group, $Aut(H)$. For a sufficiently large $n$ and a constant $\delta>0$, $\mathscr{D}_K \leq O(|K|^2e^{-\delta d})$ for subgroup, $K$.
\end{theorem}

% The change in dimension of $H$ (which is $(n-k) \times n$) from that of $M$ $(k \times n)$ is accounted for in the assumptions involved. That means, for some constant $0 < a < 1/4$, instead of the original assumption of $q^{n^2} \leq n^{an},$ we have $q^{{(n-k)}^2} \leq n^{an}$ for our case.

\begin{remark}\label{rem.indist.}
    Recall equation \eqref{indist.cond.}. The subgroup, $K$ is indistinguishable if $\mathscr{D}_K \leq \log^{-\omega(1)}|G|.$ Using the assumption $q^{{(n-k)}^2} \leq n^{an}$, we can simplify $\log|G|$ as $\log|(GL_{n-k}(\mathbb{F}_q)\times S_n)^2 \rtimes \mathbb{Z}_2| = O(\log n! + \log q^{{(n-k)}^2}) = O(n\log n).$ 
    This implies that the subgroup, $K$ is indistinguishable if $|K|^2e^{-\delta d} \leq (n\log n)^{-\omega(1)}$.
\end{remark}

The size of subgroup, $K,$ is $|K| = 2|H_0|^2$ and $|H_0| = |Aut(H)|\times|Fix(H)|,$ where $Fix(H) := \{S \in GL_{n-k}(\mathbb{F}_q): SH=H\}.$ Clearly, $|Fix(H)| = 1 \text{ or } |H_0| = |Aut(H)|.$  Therefore in order to check for indistinguishability of a hidden subgroup, $K$, we need to find the size and minimal degree of the automorphism group, $Aut(H).$ 

Recall that our parity check matrix $H$ is an $m \times m\ell$ matrix. Consider $P \in Aut(H).$ By definition, there exists a matrix, $S \in GL_m(\mathbb{F}_q)$ such that, \\
\[S^{-1}HP = H.\]
Let $H = \left[\mathcal{I}| C\right]$ such that $C = \left[G_1^{*}| G_2^{*}| \hdots| G_{\ell-1}^{*}\right].$ Therefore,
\begin{equation}\label{shp-eqn}
    S^{-1}\left[\mathcal{I}| C\right]P = \left[S^{-1}| S^{-1}C\right]P = \left[\mathcal{I}| C\right].
\end{equation}

Right multiplication by the permutation matrix, $P$ permutes the columns of a matrix, $S^{-1}H$. We get matrix, $H$. Notice that, the entries in matrix $S$ are in $\mathbb{F}_p$ and entries in matrix $C$ are in $\mathbb{F}_q,$ where $q=p^s$ for some integer $s> 0.$ Therefore, matrices $S^{-1}C$ and $C$ cannot have the same set of columns. In other words, $P$ permutes the columns of $S^{-1}$ such that we get matrix $I$. Matrix $P$ permutes columns $S^{-1}C$ such that we get $C.$ This implies that every permutation matrix $P \in Aut(H)$ is a block diagonal matrix as defined in equation \eqref{eq:perm.matrix}. The first block, $P_0$ is an $m \times m$ matrix. Given that matrix $C$ is a twistulant matrix, for all $1 \leq i \leq \ell-1$ block matrices $P_i$ are of size $m \times m$. That is,
\begin{equation}\label{eq:perm.matrix}
    P =
\begin{pmatrix}
    P_0 & & & & \\
    & P_1 & & \smash{\text{\huge 0}} & \\
    & & \ddots & & \\
    & & & P_{l-2} & \\
    & \smash{\text{\huge 0}} & & & P_{\ell-1} \\
\end{pmatrix}_{m\ell \times m\ell}
\end{equation}
where each block $P_i$ acts on $S^{-1}G_i^*,$ for all $1 \leq i \leq \ell-1.$ We can rewrite equation \eqref{shp-eqn} as follows,
\begin{equation}
    \begin{split}
        \left[S^{-1}| S^{-1}C\right]P = & \left[S^{-1}P_0| S^{-1}G_1^{*}P_1 | S^{-1}G_2^{*}P_2 | \hdots| S^{-1}G_{\ell-1}^{*}P_{\ell-1}\right] \\
        = & \left[I|G_1^{*}| G_2^{*}| \hdots| G_{\ell-1}^{*}\right]
    \end{split}
\end{equation}
Observe that $S^{-1}P_0 = I$, that is  $S^{-1} = P_0^{-1}$. Thus, $P_0^{-1}G_i^{*}P_i = G_i^{*}$ for all $1 \leq i \leq \ell-1.$
Recall that $Aut(H) = \{P \in S_n: S^{-1}HP = H \text{ for some }S \in GL_{m}(\mathbb{F}_p)\}.$ Then the size of $Aut(H)$ can rewritten as the cardinality of the following set
\begin{equation}\label{eq:A1}
    A_1 := \left\{(P_0, P_i) | P_0^{-1}G_i^{*}P_i = G_i^{*} \text{ for some } 1 \leq i \leq \ell-1 \right\}.
\end{equation}
The second condition in remark \ref{rem.twistulant} implies that no two columns in any $G_i^*$ can be identical. Together with the first condition in remark \ref{rem.twistulant}, we have that no two columns of $H$ are identical. Moreover, $S^{-1}H = P_0^{-1}H$ cannot have any identical columns. From equation \eqref{eq:A1} we have that for all $1 \leq i \leq \ell-1$, $P_0^{-1}G_i^*$ has distinct columns. Then there exists at most one $P_i$ such that we obtain $G_i^*$ by permuting the columns of $P_0^{-1}G_i^*$ using $P_i$. The cardinality of the set $A_1$ is the number of such $P_0$'s. That is, if
\begin{equation}
    A_2 := \Big\{P_0 \in S_m | P_0^{-1}G_i^{*}P_i = G_i^{*} \text{ for all } 1 \leq i \leq \ell-1 \Big\} 
\end{equation}
then, $|Aut(H)| = |A_1| = |A_2|.$ Moreover, $Aut(H)$ is a subgroup of $S_n$ and $A_2$ is a subgroup of $S_m$.

Let $GL_{n}(\mathbb{F}_q)$ be the general linear group over $\mathbb{F}_q$ (a group of $n \times n$ invertible matrices), $AFL_1(p)$ be the affine group of degree one over field, $\mathbb{F}_p$, where $p$ is a prime number. The Burnside-Schur theorem \cite[Thm. 3.5B]{Burnside.prime} (restated below) is useful in finding the size and minimal degree of $Aut(H).$
\begin{theorem}{\bf (Burnside-Schur)} 
    Let $p$ be a prime. Every primitive finite permutation group of size, $p$ containing a regular cyclic subgroup is either 2-transitive or permutationally isomorphic to a subgroup of the affine group $AGL_1(p)$. 
\end{theorem}

The result was extended to finite permutation groups of size, $p^s$ for integer $s>0$ (\cite[pp. 339-343]{Burnside}). Given $A_2$ is a subgroup of $S_m,$ where  $m$ is a prime power for some prime other than the characteristic of $\mathbb{F}_q$, if $A_2$ contains an $m$-cycle, then we can apply the Burnside-Schur theorem on $A_2.$ Consider the following matrix,
\begin{equation}
    B =
    \begin{bmatrix}
    0 & 0 & \hdots & 0 & \lambda \\
    1 & 0 & \hdots & 0 & 0 \\
    0 & 1 & \hdots & 0 & 0 \\
    \vdots & \vdots & \ddots & \vdots & \vdots \\
    0 & 0 & \hdots & 1 & 0 \\
    \end{bmatrix}_{m \times m}.
\end{equation}
Matrix $B$ corresponds to an $m$-cycle. We check if $B$ is an element of set $A_2.$ If $B \in A_2,$ then for all $1 \leq i \leq \ell-1$, we have $B^{-1}G_i^{*}P_i = G_i^{*}$. Consider the matrix $B^{-1}$,
\begin{equation}
    B^{-1} =
    \begin{bmatrix}
    0 & 1 & 0 & 0 & \hdots & 0 \\
    0 & 0 & 1 & 0 & \hdots & 0 \\
    \vdots & \vdots & & \ddots & \vdots & \vdots \\
    0 & 0 & 0 & 0 & \hdots & 1 \\
    \lambda^{-1} & 0 & 0 & 0 & \hdots & 0 \\
    \end{bmatrix}_{m \times m}.
\end{equation}
Then, for all $1 \leq i \leq \ell-1$,
\begin{equation}
\begin{aligned}
    B^{-1}G_i^{*} = &
    \begin{bmatrix}
    0 & 1 & 0 & 0 & \hdots & 0 \\
    0 & 0 & 1 & 0 & \hdots & 0 \\
    \vdots & \vdots & & \ddots & \vdots & \vdots \\
    0 & 0 & 0 & 0 & \hdots & 1 \\
    \lambda^{-1} & 0 & 0 & 0 & \hdots & 0 \\
    \end{bmatrix}
    \begin{bmatrix}
    c_0 & c_1 & c_2 & \hdots & c_{m-1}\\
    \lambda^{-1} c_{m-1} & c_0 & c_1 & \hdots & c_{m-2}\\
    \lambda^{-1} c_{m-2} & \lambda^{-1} c_{m-1} & c_0 & \hdots & c_{m-3}\\
    \vdots & \vdots & \vdots & \ddots & \vdots \\
    \lambda^{-1} c_1 & \lambda^{-1} c_2 & \lambda^{-1} c_3 & \hdots & c_0
    \end{bmatrix} \\
    = &
    \begin{bmatrix}
    \lambda^{-1} c_{m-1} & c_0 & c_1 & \hdots & c_{m-2}\\
    \lambda^{-1} c_{m-2} & \lambda^{-1} c_{m-1} & c_0 & \hdots & c_{m-3}\\
    \vdots & \vdots & \vdots & \ddots & \vdots \\
    \lambda^{-1} c_1 & \lambda^{-1} c_2 & \lambda^{-1} c_3 & \hdots & c_0 \\
    \lambda^{-1}c_0 & \lambda^{-1}c_1 & \lambda^{-1}c_2 & \hdots & \lambda^{-1}c_{m-1}
    \end{bmatrix}. \\
\end{aligned}
\end{equation}
For $P_i = B$, we show that $B^{-1}G_i^{*}P_i = G_i^{*}$.
\begin{equation}
\begin{aligned}
    B^{-1}G_i^{*}B = &
    \begin{bmatrix}
    \lambda^{-1} c_{m-1} & c_0 & c_1 & \hdots & c_{m-2}\\
    \lambda^{-1} c_{m-2} & \lambda^{-1} c_{m-1} & c_0 & \hdots & c_{m-3}\\
    \vdots & \vdots & \vdots & \ddots & \vdots \\
    \lambda^{-1} c_1 & \lambda^{-1} c_2 & \lambda^{-1} c_3 & \hdots & c_0 \\
    \lambda^{-1}c_0 & \lambda^{-1}c_1 & \lambda^{-1}c_2 & \hdots & \lambda^{-1}c_{m-1}
    \end{bmatrix} 
    \begin{bmatrix}
    0 & 0 & \hdots & 0 & \lambda \\
    1 & 0 & \hdots & 0 & 0 \\
    0 & 1 & \hdots & 0 & 0 \\
    \vdots & \vdots & \ddots & \vdots & \vdots \\
    0 & 0 & \hdots & 1 & 0 \\
    \end{bmatrix} \\
    = &
    \begin{bmatrix}
    c_0 & c_1 & c_2 & \hdots & c_{m-1}\\
    \lambda^{-1} c_{m-1} & c_0 & c_1 & \hdots & c_{m-2}\\
    \lambda^{-1} c_{m-2} & \lambda^{-1} c_{m-1} & c_0 & \hdots & c_{m-3}\\
    \vdots & \vdots & \vdots & \ddots & \vdots \\
    \lambda^{-1} c_1 & \lambda^{-1} c_2 & \lambda^{-1} c_3 & \hdots & c_0  
    \end{bmatrix} 
    = G_i^{*}\\
\end{aligned}
\end{equation}
Therefore, $B \in A_2$. Applying the Burnside-Schur theorem to $A_2,$ we can conclude that $A_2$ is either 2-transitive or is permutationally isomorphic to $AGL_1(m).$ Due to the chosen form of parity check matrix $H = \left[\mathcal{I}|c\right]$, $A_2$ cannot be 2-transitive. Therefore, $A_2 \leq AGL_1(m)$ and size of $A_2$ is bounded above by the size of $AGL_1(m)$, that is $|A_2| \leq m(m-1)$. Therefore, $|Aut(H)| \leq m(m-1).$

Given that $B \in A_2$ and $A_2 \leq AGL_1(m)$, $B \in AGL_1(m)$. For some $A \in M_{m \times m}(\mathbb{F})$ and $b \in \mathbb{F}_m$, we have $B(x) = Ax+b \mod m,$ for all $x \in \mathbb{F}_m$. If $B$ fixes more than one point, then $A$ must be the $m \times m$ identity matrix, $\mathcal{I}_m$ and $b$ must be the zero vector, $b = 0$. This implies that $B$ is an identity element. Therefore, any non-trivial element fixes at most one element and the minimal degree of $Aut(H)$ is at least $m-1$.

Theorem \ref{thm:QFS.safe} provides necessary conditions such that the hidden subgroup, $K$ corresponding to our proposed cryptosystem, is indistinguishable.

\begin{theorem}\label{thm:QFS.safe}
   Let $p$ be a prime number and $q = p^s$ for some integer $s>0$. Consider an $[n=m\ell, k=(\ell-1)m]_q$ $(\lambda,\ell)$-quasi-twisted code $\mathcal{C},$ where  $m$ is a prime power for some prime other than the characteristic of $\mathbb{F}_q$. Suppose $m < 1/4 \ell (\log_q m + \log_q \ell).$ Then the subgroup $K$ defined as follows.
    \[K = G \wr \mathbb{Z}_2|_f = \Bigg(\bigg(\Big(H_0, s^{-1}H_0s\Big),0\bigg) \bigcup \bigg(\Big(H_0s, s^{-1}H_0\Big),1\bigg)\Bigg),\]
    is indistinguishable.
\end{theorem}
\proof
Let $d$ be the minimal degree of the automorphism group $Aut(H)$ and $\delta > 0$ be a constant. From theorem \ref{thm.indist.} and remark \ref{rem.indist.}, the subgroup $K$ is indistinguishable if $q^{(n-k)^2} \leq n^{an}$ and $|K|^2e^{-\delta d} \leq (n\log n)^{-\omega(1)}$. For some $0< a < 1/4$, we have $q^{(n-k)^2} \leq n^{an}$. 
From here on, $\log$ refers to the logarithm with base $q$ unless stated otherwise. Then, 
\vspace{-2mm}
\begin{equation}
\begin{aligned}
    q^{(n-k)^2} = q^{m^2} & \leq (m\ell)^{am\ell}\\
    m^2 & \leq am\ell\log(m\ell)\\
    m^2 & \leq am\ell(\log m + \log \ell)\\
    m & \leq al(\log m + \log \ell). \\
\end{aligned}
\end{equation}
For $0<a<1/4$, we have $m < 1/4 \ell (\log_q m + \log_q \ell)$. 

We use theorem \ref{thm.indist.} to check for indistinguishability of $K$. Recall that $|K| = 2|H_0|^2 = 2|Aut(H)|^2$, where $|Aut(H)| \leq m(m-1)$. For $d \geq m-1$, and constant $\delta>0$, substituting values for $|K|$ and $d$, we get 
\begin{equation}
    |K|^2e^{-\delta d} \leq (2m^4)^2e^{-\delta m}.
\end{equation}
%\vspace{-2mm}
Therefore, $|K|^2e^{-\delta d} \leq 4m^8e^{-\delta m}$. Using the bound on $\ell,$ we can see that $4m^8e^{-\delta m} \leq (m\ell\log (m\ell))^{-\omega(1)}.$ Hence, 
$|K|^2e^{-\delta d} \leq (m\ell\log (m\ell))^{-\omega(1)}$. That is, the hidden subgroup $K$ is indistinguishable.
\qed
\end{proof}

Since the hidden subgroup $K$ corresponding to our cryptosystem is indistinguishable, our cryptosystem is secure against QFS based attacks. 

%-----------------------------------------------------------------------------------
\section{Conclusion}

%\RK{Include why we don't have a numerical analysis. Needs to be polished.}

This paper addresses the rising need for a quantum-secure cryptosystem in response to the expanding capabilities of quantum computing. The McEliece cryptosystem (equivalently, the Niederreiter cryptosystem) is a potential solution, given that it is quantum-secure and facilitates faster encryption and decryption. However, the large key sizes make it not feasible for implementation. We analyzed different modifications of the McEliece/Niederreiter cryptosystems and found a variant of the Niederreiter cryptosystem based on quasi-cyclic codes (presented in \cite{Upen-Pap}). This work inspired us to further study the quasi-cyclic codes and a more generalized code family - quasi-twisted codes. We explored the possibility of replacing Goppa codes in the Niederreiter cryptosystem with quasi-twisted codes to develop an implementable as well as secure code-based cryptosystem.

It is necessary for a code to have an efficient decoding algorithm in order to adopt it in a code-based cryptosystem. One main challenge in developing a  Niederreiter-like cryptosystem based on quasi-twisted codes was the lack of an efficient decoding algorithm. In this work, we propose a novel decoding algorithm for QT codes. In order to provide a decoding procedure, we show a new lower bound on the minimum distance of quasi-twisted codes - a Hartmann-Tzeng(HT)-like bound, $d^*$ (theorem \ref{HT-like Bound-Th}). We then describe a syndrome-based decoding algorithm that can correct up to $\varepsilon = \frac{d^* -1}{2}$ errors (section \ref{subsec:dec-algo}, table \ref{tab:dec.algo}). Our algorithm operates with a time complexity that is quadratic in relation to the length of the code.

Having described an efficient decoding procedure for quasi-twisted codes, we present a Niederreiter-like cryptosystem based on quasi-twisted code (section \ref{subsec:cryptosystem.desc}). We draw motivation from the quasi-cyclic variant (\cite{Upen-Pap}) and develop a cryptosystem that uses $[m\ell, (\ell-1)m]_q$ $(\lambda,\ell)$-quasi-twisted codes, where $m =q^r$ is a prime power and $\ell$ is bounded above by a polynomial in $m.$ We show that our cryptosystem is secure against some classical attacks. To prove quantum security, we consider the scrambler permutation attack and prove that our cryptosystem can resist such an attack. This is done by proving the indistinguishability of the hidden subgroup in the hidden subgroup problem, by quantum Fourier sampling (QFS). To the best of our knowledge, the quasi-cyclic code-based cryptosystem is not proven to be insecure against any known attack. Therefore, we expect that this would be the case for our cryptosystem as well. However, we only prove that our cryptosystem is secure against QFS-based quantum attacks. We leave the analysis of the security of our cryptosystem against attacks that are not based on QFS, such as variants of the  Sidelnikov-Shestakov attack \cite{sidelnikov.GRS.attack} that uses the dimension of the Schur square of the underlying code for future work. Furthermore, the parameter bounds presented in this work may not represent the most optimal values. There is potential for these bounds to be further refined and improved through more in-depth analysis in future research, which could lead to enhanced performance and more accurate results. Currently, we are unable to perform a numerical analysis of our cryptosystem and its security given the lack of existing examples of quasi-twisted codes with the required structure and leave the analysis to a later date depending on development of required codes.

 \section*{Acknowledgment}
 The authors are grateful to Dr. Krishna Kaipa and Dr. Upendra Kapshikar for useful comments and discussions. They express their gratitude to Prof. N. Aydin for his valuable suggestions and comments. They also acknowledge Dr. Simona Samardjiska for her useful insights.

\end{document}